\documentclass[11pt,draftclsnofoot,onecolumn]{IEEEtran}

\usepackage{graphicx}
\usepackage{amsmath}  % IEEE no longer uses Computer Modern fonts, so avoid [cmex10]
\usepackage{amsfonts}
\usepackage{algorithmicx}
\usepackage{algpseudocode}
\usepackage{algorithm}
\usepackage{color}
\usepackage{verbatim}
\usepackage{graphicx,wrapfig,lipsum}
\usepackage[colorinlistoftodos]{todonotes}
\usepackage{amssymb}
\usepackage{amsthm}
\usepackage{comment}
\usepackage{mathtools}
\usepackage{bm}
\usepackage{dsfont}
 \usepackage{cite}
\graphicspath{{fig/}}

\usepackage{url}

\newcommand{\ep}{\frac{\xi}{6}}
\newcommand{\epp}[1]{\frac{\xi}{#1}}
\newcommand{\epq}[2]{\frac{#1\xi}{#2}}
\newcommand{\eps}[1]{\frac{\xi^2}{#1}}



\newtheorem{theorem}{Theorem}
\newtheorem{lemma}[theorem]{Lemma}
\newtheorem{proposition}[theorem]{Proposition}
\newtheorem{definition}[theorem]{Definition}

\title{On Analysis of the Bitcoin and Prism Backbone Protocols}
\author{Jing Li and Dongning Guo\\
Northwestern University\\
jingli2015@u.northwestern.edu, dGuo@Northwestern.edu\\
\today}

\begin{document}
\maketitle
\begin{abstract}
Bitcoin is a peer-to-peer payment system proposed by Nakamoto in 2008.  Properties of the bitcoin backbone protocol have been investigated in some depth: the blockchain growth property quantifies the number of blocks added to the blockchain during any time intervals; the blockchain quality property ensures the honest miners always contribute at least a certain fraction of the blockchain; the common prefix property ensures if a block is deep enough, it will eventually be adopted by all honest miners with high probability.
Following the spirit of decoupling various functionalities of the blockchain, Bagaria, Kannan, Tse, Fanti, and Viswanath (2018) proposed the Prism protocol to dramatically improve the transaction rate  while maintaining the same level of security.
Most prior analyses of the bitcoin and the Prism backbone protocols provide performance guarantees up until a finite number of rounds (equivalent to finite lifespan) and that all miners have identical information by the end of each round (referred to as the synchronous model).   This paper presents a streamlined and strengthened analysis without the finite lifespan assumption.  Also, both the synchronous model and a more general model with arbitrary but bounded block propagation delays are studied.
The results include a blockchain growth property, a blockchain quality property, and a common prefix property of the bitcoin backbone protocol, as well as the liveness and persistence of the Prism backbone protocol.
An explicit probabilistic guarantee is provided for every transaction found in an honest blockchain to become permanent in the final ledger.
The properties of the bitcoin and the Prism backbone protocols are given as explicit expressions rather than order optimal results, which lead to improved references for public transaction ledger protocol design.
\end{abstract}

\section{Introduction}
\subsection{The bitcoin backbone protocol}
Bitcoin is an electronic payment system introduced by Nakamoto\cite{nakamoto2008bitcoin} in 2008. The system is built on a distributed ledge technology commonly referred to as blockchain. Miners are distributed parties who generate blocks and maintain their own version of the blockchain.
A blockchain is a finite sequence of blocks adopted by some miner at some point in time. It begins with a genesis block, and every subsequent block contains a cryptographic hashing of the previous block.
In order to generate a valid new block, a miner need to find a nonce whose hash values {\color{black}satisfies a difficulty requirement}. The process of finding such a nonce is called mining.
An honest miner follows the honest chain rule, i.e., it always adopts the longest blockchain it heard about and mines on top of the longest blockchain.
Since all miners work simultaneously, it is possible that two {\color{black}or more} different blocks are mined and announced at around the same time. Then {\color{black}different honest miners may extend different blockchains depending on which longest one they hear first}. This phenomenon is called forking. Forking of a blockchain challenges network consensus and presents opportunities for double spending attack, namely, a transaction included in the longest fork is not included in a different fork that overtakes the first fork to become the longest one.

%Double spending is an attack where a given set of currency is spend in more than one transaction. Main ways to perform a double spending attack includes race attack, Finney attack, $50\%$ attack, etc..
%The scalability and security of the bitcoin protocol are widely investigated.
%The ledger is maintained by distributed parties called miners who run the bitcoin protocol.
%An important security concern is double spending. The adversary first post a transaction into the blockchain. After receiving goods or services, the adversarial manage to alter the blockchain and revoke the previous transaction.
 Nakamoto\cite{nakamoto2008bitcoin} characterized the race between the honest {\color{black}miners and an adversary with less than half of the total mining power} as a random walk with a drift. {\color{black}Nakamoto} showed that the probability the adversary blockchain overtakes the honest miner's consensus blockchain vanishes exponentially {\color{black}over time}. Nakamoto argued that the bitcoin protocol is safe under double spending attack {\color{black}as long as one considers a transaction confirmed only after enough new blocks are mined to extend the honest blockchain}.
{\color{black}An in-depth} analysis of the bitcoin protocol was given in \cite{narayanan2016bitcoin}.
{\color{black}Several important properties of the bitcoin backbone protocol have been proposed} in \cite{nakamoto2008bitcoin,garay2015bitcoin, kiayias2015speed,sapirshtein2016optimal, garay2017bitcoin}.
Garay, Kiayias, and Leonardos\cite{garay2015bitcoin} gave a formal description and analysis of the bitcoin backbone protocol assuming a fully synchronous network, namely, mining takes place in rounds and at the end of each round, all miners see all published blocks. Under this model, \cite{garay2015bitcoin} introduced a common prefix property and a blockchain  quality property. The common prefix property states if a block is $k$ blocks deep in an honest miner's blockchain, then the probability that the block is not included by all other honest miners' blockchain decreases exponentially with $k$. The blockchain  quality property states the honest miners always contribute at least a certain percentage of the blockchain regardless of the strategy of adversarial parties.
Then, \cite{kiayias2015speed} introduced a blockchain growth property, which quantifies the number of blocks added to the blockchain during any time intervals.

Moreover, Nakamoto's analysis was improved in \cite{sapirshtein2016optimal} to address selfish mining. In this case, selfish miners can introduce disagreement between  honest miners and split their hashing power.
Selfish miners thus enhance their relative hashing power {\color{black}to win disproportionate rewards}.
This strategy, however, is not designed for double spending purposes.

The bitcoin backbone protocol gives birth to numerous ``robust public transaction ledger'' protocols\cite{eyal2016bitcoin, wood2014ethereum, bagaria2018deconstructing}. The preceding properties guarantee two fundamental properties of a robust public transaction ledger: liveness and persistence. Due to the  blockchain  growth property and the  blockchain  quality property, blocks originating from honest miners will eventually end up at a level of more than $k$ blocks of an honest miner's blockchain. Due to the common prefix property, an honest miner's $k$-deep block remains permanent.

%These properties are fundamental for arguing the security of a robust transaction ledger.

% Liveness says blocks originating from honest miners will eventually end up at a level of more than $k$ blocks of an honest miner's blockchain regardless of denial of service attacks by the adversary.

The bitcoin backbone protocol can also be leveraged to solve other problems. For example, the bitcoin backbone protocol ensures some basic properties for some randomized Byzantine agreement protocols\cite{gramoli2017blockchain, miller2014anonymous, decker2016bitcoin, lamport1982byzantine, feldman1988optimal}.

%Due to network delay or failure, a strictly synchronous network is hard to achieve in practice. It is necessary to consider the performance of the bitcoin backbone protocol in asynchronous or partially synchronous networks. By saying partially synchronous networks, we mean bounded-delay networks where messages delivery can be delayed, but the delay is upper bounded.
%It has been shown in \cite{pass2017analysis} that in a fully asynchronous network, a blockchain consensus mechanism with consistency and liveness is impossible. This negative result is consistent with the experimental result in \cite{decker2013information}. However, the  blockchain  growth property, the  blockchain  quality property, and the common prefix property remain valid in asynchronous networks with bounded delays\cite{garay2015bitcoin, pass2017analysis}.

\subsection{The Prism protocol}
The throughput of bitcoin is limited {\color{black}by design  to ensure security}\cite{decker2013information}. As mining rate increases, blocks are more likely to be mined and announced simultaneously, i.e., forking is more likely to occur. Due to the longest  blockchain  rule, only the blocks on the longest blockchain will eventually be adopted by honest miners, and other honest blocks are wasted. Then the adversarial miners {\color{black}compete with fewer honest miners. To avoid forking, the average time interval between new blocks is set to be much longer than the latency for propagating a block to most miners in the network}\cite{sompolinsky2015secure}.

{\color{black}Many ideas have been} proposed to improve the blockchain throughput {\color{black}while} maintaining its security. One way is to deal with high-forking blockchains by optimizing the forking rule. For example, GHOST chooses the main blockchain according to the heaviest tree rule instead of the longest  blockchain  rule\cite{sompolinsky2015secure}. Inclusive, Spectre, and Phantom construct a directed acyclic graph (DAG) structured blockchain by introducing reference links between blocks in addition to the parent links\cite{lewenberg2015inclusive, sompolinsky2016spectre, sompolinsky2018phantom}. However, these protocols are vulnerable to certain attacks\cite{natoli2016balance, li2018scaling,zheng2016blockchain}. Generally speaking it is very challenging to make high-forking protocols secure.

Another line of work is to decouple the various functionalities of the blockchain.
%To ensure persistence and liveness, a blockchain need to 1) include new transactions and 2) finalize the order of transactions. In the bitcoin protocol, each block serves these two functions simultaneously. Each block includes new transactions and mine on the top of the longest blockchain to decide the main blockchain. The main blockchain decides the order of transactions.
For example, BitcoinNG divides the bitcoin blockchain's operations into leader selection and transaction serialization\cite{eyal2016bitcoin}. In BitcoinNG, time is divided into epochs. During each epoch, a leader is chosen to order the transaction blocks of that epoch. However, this protocol is vulnerable to bribery or targeted attacks to leaders. In Fruitchain, transactions (fruits) are also decoupled from proposer blocks. However, fruitchain focuses on  enhancing  fairness instead of improving throughput\cite{pass2017fruitchains}.

Following the spirit of decoupling blocks' functionalities,
Bagaria, Kannan, Tse, Fanti, and Viswanath\cite{bagaria2018deconstructing} proposed the Prism protocol, which is a structured-DAG blockchain with one proposer blockchain and many voter blockchains. The voter blocks elect a leader block at each level of the proposer blockchain by voting. The sequence of leader blocks concludes the contents of all voter blocks, and finalizes the ledger. Each voter blockchain mines independently at a low mining rate. A voter blockchain follows the bitcoin protocol to provide security to leader election process.

With this design, the throughput (containing the content of {\em all} voter blocks) is decoupled from the mining rate of each voter blockchain. Slow mining rate guarantees the security of each voter blockchain as well as the proposer blockchain.
Prism achieves security against up to 50\% adversarial hashing power, optimal throughput up to the capacity of the network, and fast confirmation latency for honest transactions. A thorough description and analysis is shown in \cite{bagaria2018deconstructing}.

\subsection{Our results}
Previous analysis on the backbone of bitcoin and Prism assumes a blockchain's lifespan is finite, i.e., there exists a maximum round when the blockchain ends. For example, in \cite{garay2015bitcoin,garay2017bitcoin} and \cite{bagaria2018deconstructing},  the good properties of blockchain hold only under typical events, i.e., the number of honest and adversarial blocks mined  must not deviate too much from their expected value over all long enough time intervals. The probability of typical events was shown to depend on the blockchain's maximum round parameter. Indeed,
%In partially synchronous networks, \cite{pass2017analysis} also calculates the probability of major properties by taking union bounds over all rounds.
the probability of the  blockchain  growth property, the  blockchain  quality property, and the common prefix property are all expressed implicitly in terms of the blockchain's maximum round.

In this paper, we drop the finite horizon assumption and prove strong  properties of the bitcoin backbone protocol. We define the typical events with respect to each interval: instead of requiring the number of honest and adversarial blocks to be typical over all long enough time intervals, we only require them to be typical over all time intervals that contain a certain interval that includes the transaction of interest. Since the probability that the number of honest and adversarial blocks are ``atypical'' decreases exponentially with interval length,
the sum of the probabilities over all those intervals remains vanishingly small. Thus we provide performance guarantees that are truly {\em permanent} whether or not the blockchain have a finite lifespan.
Moreover, without the finite horizon assumption, we express the properties of the bitcoin backbone protocol in explicit  expressions in lieu of order optimality results in some previous analysis. The explicit expressions provide tighter bounds and more practical references to public transaction ledger protocol design.

%Instead of requiring the number of honest and adversarial blocks to be typical over all long enough time intervals, we require the number of honest and adversarial blocks to be typical only for all time intervals containing the bitcoin backbone protocol.
%Since the probability that an interval is  ``untypical'' decreases exponentially with interval length, summing them up will be upper bounded by a constant.
%Thus, the probability of our typical events can be expressed without the maximum round. With modifications of the proof, the properties of the bitcoin protocol an be proved without concerning the lifetime of blockchain.
%expressions for parameters to ensure the security instead of big $O$ notation proposed in previous analysis.

%Currently, liveness and consistency of the Prism protocol is proved only in synchronous networks with the maximum length assumption\cite{bagaria2018deconstructing}.
In \cite{bagaria2018deconstructing}, liveness and consistency properties of the Prism protocol were proved assuming a finite life span of the blockchains\cite{bagaria2018deconstructing}.
In this paper, we also prove the liveness and consistency of the Prism protocol without the finite horizon assumption. %and in bounded-delay networks.
%In both synchronous and partially synchronous networks, we give explicit expressions without the maximum rounds parameter to ensures security under a fixed security parameter.

Anothe crucial assumption in~\cite{
garay2015bitcoin,
garay2017bitcoin,
bagaria2018deconstructing
}
is that all blocks broadcast during a protocol round reach all miners by the end of that round, i.e., all miners have complete up-to-date information by the end of each round.  This is referred to as the synchronous model.  In this paper, we generalize the analysis to a much more challenging model in which a block may reach different miners after arbitrary different delays, so that even the honest miners are never guaranteed to have identical view of the system.  It is only assumed that the propagation time is bounded by $T$ rounds, which is realistic in practice.  A key idea in this paper is to exploit honest miners' common information about those rounds in which a single honest block is mined and that no other honest blocks are mined within $T-1$ rounds before and after.  Essentially all the properties developed for the synchronous model find their counterparts for this bounded-delay model.
% examine the so-called $T$-isolated successful rounds

\section{Model and definitions} \label{sec: model and definitions}
We assume the total number of miners is $n$, among which $t$ miners are adversarial and the remaining miners are honest. Assume all miners have equal hash powers (if not, we assume they can be split into equal-power pieces).
Let
\begin{align}
    \beta = \frac{t}{n}
\end{align}
denote the percentage of adversarial miners. We assume adversarial miners collectively have less than $\frac{1}{2}$ of the total mining power in the blockchain network, so $\beta \in [0,\frac{1}{2})$.

We adopt a discrete model where activities take place in rounds.
If a miner publishes one or more blocks in a round, all miners receive the block(s) at exactly the end of the round (a miner can only react to round $r$ blocks in round $r+1$).
Evidently, by the end of each round, all honest miners are fully synchronized.
If a block is mined by an honest miner, we call it an {\em honest block}; otherwise the block is called an {\em adversarial block}.
We assume that during round $0$, a single honest block, called  the genesis block, is mined and broadcast to all miners. For $r\in \{1,2,\ldots\}$,
let $H[r]$ denote the number of all honest blocks mined during round $r$.
The mining difficulty and miner' mining powers are adjusted to be constant in all rounds $r\ge 1$.

Without loss of generality, the mining power of all miners are and the mining difficulty are assumed to remain constant, such that
the probability that an honest miner mines a new block in every round $r\ge 1$ is equal to $p\in (0,1)$.\footnote{This probability is held constant by adjusting the mining difficulty in case the mining power fluctuate over rounds.}
 Note that $H[r]\sim Binomial(n-t, p)$.  Define
\begin{align}
    X[r] =
    \begin{cases}
    1, \; \; & \text{if}\; H[r] \ge 1\\
    0, & \text{otherwise}.
    \end{cases}
\end{align}
$X[r]$ indicates if one or more honest blocks are mined during round $r$ {\color{black}or not}. Let
\begin{align}
    q = 1 - (1-p)^{n-t}.
\end{align}
Then $X[r]\sim Bernoulli(q)$.
Define
\begin{align}
    Y[r] =
    \begin{cases}
    1, \; \; & \text{if}\; H[r] = 1\\
    0, & \text{otherwise}.
    \end{cases}
\end{align}
Basically $Y[r]$ indicates if a single honest block is mined in round $r$ {\color{black}or not}. Then  $Y[r]\sim Bernoulli((n-t)p(1-p)^{n-t-1})$.  A round $r$ is called a uniquely successful round if $Y[r] = 1$.
Let $Z[r]$ upper bound the number of adversarial blocks mined during round $r$ (the adversarial miners may or may not publish them). Then $Z[r]\sim Binomial(t,p)$.

It is important to note that $H[1], H[2], \ldots$ are independently and identically distributed (i.i.d.), which form a stationary process. The same can be said of the $X$, $Y$, and $Z$ sequences. Define
\begin{align}
\xi & = \frac{1-2\beta}{1-\beta}. \label{def: xi}
\end{align}
Then $\xi \in (0,1]$.

For all integers $s$ and $r$ satisfying $1\le s < r$, let
\begin{align}
    H[s,r] = \sum_{i=s}^{r-1}H[i],
\end{align}
which represents the total number of honest blocks mined during rounds $s,\ldots,r-1$. To be consistent with this notation, we mean all rounds up to and including $r-1$ when we say ``by round $r$''. Likewise, we define
\begin{align}
    X[s,r] &= \sum_{i=s}^{r-1}X[i]\\
    Y[s,r] &= \sum_{i=s}^{r-1}Y[i]\\
    Z[s,r] &= \sum_{i=s}^{r-1}Z[i].
\end{align}

\begin{definition}
By a {\em blockchain} we mean a finite sequence of blocks  adopted by some miner at some point in time which begins with a genesis block and that every subsequent block contains a cryptographic hashing of the previous block. It is assumed that no block can be mined in an earlier round than its immediate predecessor.
\end{definition}

A blockchain's prefix is also a blockchain. A blockchain must have the following properties: 1) Its blocks must be mined in order; 2) it is immutable in the sense that it is computationally impossible  for {\color{black}any} miner to mine a different blockchain that has the same genesis block and the same final block.

\begin{definition}
If a blockchain  is adopted by an honest miner by {\color{black}some} round, it is said to be honest.
\end{definition}

\section{The bitcoin backbone protocol} \label{sec: bitcoin backbone}
In Section \ref{sec: bitcoin backbone} and \ref{sec: prism} , it is assumed that, the mining difficulty is adjusted such that
\begin{align} \label{equ: honest majority}
    q \le \frac{\xi}{6}.
\end{align}
We will make heavy use of Bernoulli's  inequality:
\begin{proposition}
(Bernoulli's inequality) For every integer $k\ge 0$ and real number $x > -1$,\begin{align} \label{equ: bernoulli}
    (1+x)^k \ge 1+kx.
\end{align}
\end{proposition}

\begin{proposition} \label{prop: X}
For $r = 1,2,\ldots$ ,
\begin{align}
 q \le p(n-t) < \frac{q}{1-q}. \label{0.01}
\end{align}
\end{proposition}
\begin{proof}
As $X[r]\sim Bernoulli(q)$, we have
\begin{align}
    \mathbb{E}[X[r]] & = q \label{equ: E[X]=q} \\
    & = 1-(1-p)^{n-t} \\
    & \le p(n-t), \label{equ: q<p(n-t)}
\end{align}
where \eqref{equ: q<p(n-t)} is due to Bernoulli's inequality. Moreover,
\begin{align} \label{equ: q/1-q}
\frac{q}{1-q} & = \frac{1-(1-p)^{n-t} }{(1-p)^{n-t} } \\
& = (1-p)^{-(n-t)}-1 \\
& > (1+p)^{n-t}-1 \label{0.1}\\
& \ge p(n-t),\label{0.2}
\end{align}
where \eqref{0.1} is due to $(1+p)(1-p)<1$ and \eqref{0.2} is due to Bernoulli's inequality. By \eqref{equ: q<p(n-t)} and \eqref{0.2},
\begin{align}
 q \le p(n-t) < \frac{q}{1-q}.
\end{align}
\end{proof}

\begin{proposition} \label{prop: Y}
For $r = 1,2,\ldots$ ,
\begin{align}
    \mathbb{E}[Y[r]] > q(1-q).
\end{align}
\end{proposition}
\begin{proof}
According to Proposition \ref{prop: X}, $q\le \frac{1}{6}$ implies $q < p(n-t)<\frac{1}{5}$. {\color{black}Hence,}
\begin{align} \label{equ: E[Y] > q(1-q)}
    \mathbb{E}[Y[r]] & = p(n-t)(1-p)^{n-t-1} \\
    & \ge p(n-t)(1-p(n-t-1)) \label{he: 1.0} \\
    & > p(n-t)(1-p(n-t)) \\
    & > q(1-q),  \label{he: 1.1}
\end{align}
where \eqref{he: 1.0} is due to Bernoulli's inequality, and  \eqref{he: 1.1} holds because the function $x(1-x)$ is increasing on $[0,\frac{1}{2}]$.
\end{proof}

\begin{proposition} \label{prop: Z<X}
For $r = 1,2,\ldots$ ,
\begin{align}
    \mathbb{E}[Z[r]] < \mathbb{E}[X[r]].
\end{align}
\end{proposition}
\begin{proof}
Since $Z[r]\sim Binomial(t,p)$,
\begin{align}
    \mathbb{E}[Z[r]] & = pt \\
    & = \frac{t}{n-t}p(n-t) \\
    & < \frac{t}{n-t}\frac{q}{1-q} \label{he: 3.-1}\\
    & = (1-\xi)\frac{1}{1-q}q \label{he: 3.15} \\
    & {\color{black}\le \frac{1-\xi}{1-\ep}q}\\
    & < \mathbb{E}[X[r]] \label{he: 3.2},
\end{align}
where \eqref{he: 3.-1} is due to Proposition \ref{prop: X} and \eqref{he: 3.2} is due to $q\le \frac{\xi}{6}$.
\end{proof}

\begin{definition} \label{def: E}
For all integers $1\le s < r$, define event
\begin{align}
    E[s,r] := E_1[s,r] \cap E_2[s,r] \cap E_3[s,r]
\end{align}
where
\begin{align}
    E_1[s,r] & :=  \left\{ (1-\ep)\mathbb{E}[X[s,r]] < X[s,r] < (1+\ep)\mathbb{E}[X[s,r]] \right\}\label{equ: E1} \\
    E_2[s,r] & := \left\{(1-\ep)\mathbb{E}[Y[s,r]] < Y[s,r] \right\}\label{equ: E2} \\
    E_3[s,r] & := \left\{Z[s,r] <\mathbb{E}[Z[s,r]] + \ep \mathbb{E}[X[s,r]]\right\}. \label{equ: E3}
\end{align}
\end{definition}
Under event $E_1[s,r]$, {\color{black} the number of rounds with honest block mined, $X[s, r]$,} does not deviate from its expected value by more than a fraction of $\ep$.
Under event $E_2[s, r]$,
the number of uniquely successful rounds $Y[s, r]$ is no less than $1-\ep$ of its expected value.
Under event $E_3[s, r]$,
the upper bound for
the number of adversarial blocks is no more than its expected value plus $\ep$ of the expectation of $X[s, r]$. Intuitively, under $E[s, r]$, we have 1) a ``typical'' number of rounds during which  at least one honest block is mined, 2)``enough'' uniquely successful rounds, and 3) the total number of adversarial blocks is limited.

\begin{proposition} \label{prop: Chernoff bound}
(Chernoff bound, \cite[page 69]{mitzenmacher2017probability}) Let $X\sim binomial(n,p)$. Then for every $\eta\in (0,1]$,
\begin{align} \label{equ: chernoff X<}
    P(X\le (1-\eta)pn) \le e^{-\frac{\eta^2pn}{2}},
\end{align} and
\begin{align} \label{equ: chernoff X>}
    P(X\ge (1+\eta)pn) \le e^{-\frac{\eta^2pn}{3}}.
\end{align}
\end{proposition}

Define
\begin{align}
    \eta = \eps{180}q. \label{def: gamma}
\end{align}

\begin{lemma}\label{lemma: prob of E}
For all integers $1\le s < r$,
\begin{align}
    P(E[s,r])>1-4e^{-\eta(r-s)},
\end{align}
where $\eta$ is given in \eqref{def: gamma}.
\end{lemma}
\begin{proof}
{\color{black}We first analyze events $E_1$, $E_2$, and $E_3$ separately. We have}
\begin{align}
    P(E_1[s,r]^c) & = P\left(|X[s,r]-\mathbb{E}[X[s,r]]|\ge \ep \mathbb{E}[X[s,r]]\right) \\
    & = P\left( X[s,r]\ge \mathbb{E}[X[s,r]] +\ep \mathbb{E}[X[s,r]]\right) + P\left(X[s,r]\le \mathbb{E}[X[s,r]] - \ep \mathbb{E}[X[s,r]]\right)\\
    & \le 2e^{-\eps{108}q(r-s)} \label{1.09},
\end{align}
where \eqref{1.09} is due to Proposition \ref{prop: Chernoff bound}.

Also,
\begin{align}
    P(E_2^c[s,r]) & = P\left(Y[s,r] \le (1-\ep)\mathbb{E}[Y[s,r]]\right) \\
    & \le e^{-\eps{72}\mathbb{E}[Y[s,r]]} \label{1.091}\\
    & \le  e^{-\eps{72}(1-q) q(r-s)} \label{1.10} \\
    & < e^{-\eps{72}(1-\ep) q(r-s)}, \label{1.11}
\end{align}
where \eqref{1.091} is due to Proposition \ref{prop: Chernoff bound}, \eqref{1.10} is due to Proposition \ref{prop: Y}, and \eqref{1.11} is due to $q\le \ep$.

Note that the moment generating function for binomial random variable $Z[r]\sim Binomial(t,p)$ is $(1-p+pe^u)^{t}$ (page $39$ in \cite{das1989statistical}). We have
\begin{align}
P(E_3^c[s,r])& = P\left(Z[s,r] \ge \mathbb{E}[Z[s,r]] + \ep \mathbb{E}[X[s,r]]\right)  \\
& \le P\left(Z[s,r] \ge \mathbb{E}[Z[s,r]] + \epp{12} \mathbb{E}[Z[s,r]] + \epp{12} \mathbb{E}[X[s,r]]\right)\label{2.-15} \\
& <  \frac{\mathbb{E}\left[e^{Z[s,r]u}\right]}{e^{(1+\epp{12})\mathbb{E}[Z[s,r]]u + \epp{12}\mathbb{E}[X[s,r]]u}} \label{2.-1} \\
& = \frac{(1-p+pe^u)^{t(r-s)}}{e^{(1+\epp{12})(r-s)tpu + \epp{12}(r-s)qu}} \label{he: 2.0}\\
& \le e^{\left(e^u-1-u(1+\epp{12})\right)tp(r-s)-\epp{12}qu(r-s)}, \label{he: 2.1}
\end{align}
where \eqref{2.-15} is due to Proposition \ref{prop: Z<X}, {\color{black} \eqref{2.-1} holds for all $u\ge 0$ due to  Chernoff's inequality,} and \eqref{he: 2.1} is due to $1+x \le e^x$ for every $x\ge 0$ (here $x = p(e^u-1)$). {\color{black}Pick} $u=\log(1+\epp{12})$. Then
\begin{align}
P(E_3^c[s,r])& \le e^{\left(\epp{12} - (1+\epp{12})\log(1+\epp{12})\right)tp(r-s)-\epp{12}\log(1+\epp{12})q(r-s)} \\
& <  e^{-\epp{12}\log(1+\epp{12})q(r-s)} \label{2.11} \\
& < e^{-\eps{180}q(r-s)} \label{2.12}
\end{align}
where \eqref{2.11} is due to  $(1+x)\log(1+x)> x$ for all $x> 0$, and \eqref{2.12} is due to $\log(1+\epp{12}) > \epp{15}$ for all $0 < \xi \le 1$.

Thus,
\begin{align}
P(E[s,r]) & = 1-P(E^c[s,r]) \\
& \ge 1- P(E_1^c[s,r]) - P(E_2^c[s,r]) - P(E_3^c[s,r]) \\
& > 1-4e^{-\eta(r-s)}  \label{2.13}
\end{align}
where $\eta$ is defined in \eqref{def: gamma}{\color{black},} \eqref{2.13} is due to $\eps{72}(1-\ep) >  \eps{180}$ and $\eps{108} > \eps{180}$.
\end{proof}

%------------------------------------------------------

\begin{lemma} \label{lemma: good properties}
(Typical properties lemma) For all integers $1\le s < r$, under event $E[s,r]$, the following holds.
\begin{align}\label{equ: X < (1+1/6del)qr}
(1-\ep)q(r-s) < X[s,r]  < (1+\ep)q(r-s)
\end{align}
\begin{align} \label{equ: Y > (1-1/del)qr}
Y[s,r] > (1-\epp{3})q(r-s)
\end{align}
\begin{align}\label{equ: Z < (1-2/3del)qr}
Z[s,r] <  (1-\epq{2}{3})q(r-s)
\end{align}
\begin{align} \label{equ: Z < X2}
 Z[s,r]   < (1-\epp{2})X[s,r]
\end{align}
\begin{align} \label{equ: Z < Y}
Z[s,r] < Y[s,r].
\end{align}
\end{lemma}

\begin{proof}
Under $E[s,r]$, \eqref{equ: X < (1+1/6del)qr} follows directly from \eqref{equ: E1}.

To prove \eqref{equ: Y > (1-1/del)qr},
\begin{align}
    Y[s,r] & > (1-\ep)q(1-q)(r-s) \label{3.91}\\
    & > (1-\ep)^2q(r-s) \label{3.92}\\
    & > (1-\epp{3})q(r-s),
\end{align}
where \eqref{3.91} is due to Proposition \ref{prop: Y} and \eqref{3.92} is due to $q\le \ep$.

To prove \eqref{equ: Z < (1-2/3del)qr}, we have
\begin{align}
    Z[s,r] & <  E[Z[s,r]] + \ep E[X[s,r]] \label{4.0}\\
    & \le (1-\xi)\frac{q}{1-q}(r-s) + \ep q(r-s) \label{4.1}\\
    & < (1-\epq{2}{3})q(r-s) \label{4.2}
\end{align}
where \eqref{4.0} is due to \eqref{equ: E3}, \eqref{4.1} is due to \eqref{he: 3.15}, and \eqref{4.2} is due to $q\le \ep$.

To prove \eqref{equ: Z < X2}, we have
\begin{align}
     Z[s,r] & < (1-\epq{2}{3})q(r-s) \label{6.0} \\
     & < \frac{1-\epq{2}{3}}{1-\ep} X[s,r]  \label{6.1}\\
     & < (1-\epp{2})X[s, r], \label{6.2}
\end{align}
where \eqref{6.0} is due to \eqref{4.2} and \eqref{6.1} is due  to \eqref{equ: X < (1+1/6del)qr}.

{\color{black}The inequality \eqref{equ: Z < Y} is straightforward by \eqref{equ: Z < (1-2/3del)qr} and \eqref{equ: Y > (1-1/del)qr}.}
\end{proof}

%-----------------------------------------------------

\begin{definition}
(Typical event) For all integers $1\le s < r$, define {\color{black}the} typical event {\color{black}with respect to $[s,r]$ as}
\begin{align}
    G[s, r] := \cap_{0\le a < s, b\ge 0}E[s-a,r+b].
\end{align}
\end{definition}
{\color{black}The event} $G[s, r]$ occurs {\color{black}when the events} $E[s-a,r+b]$ simultaneously {\color{black}occurs} for all $a, b$, i.e.,  the ``$E$'' {\color{black}events occur} over all intervals that contain $[s,r]$. {\color{black}The event $G$ represents a collection of outcomes that constrain the number of blocks mined in all intervals that contain $[s,r]$, including arbitrarily large intervals that terminate in the arbitrarily far future.}
Intuitively, we {\color{black}have defined} $G[s, r]$ {\color{black}to allow} the ``good'' properties mentioned in Lemma \ref{lemma: good properties} to extend to all intervals containing $[s, r]$ {\color{black}under the event}.
%%[DG
It is important to note that the typical events defined in~\cite{
garay2015bitcoin,
bagaria2018deconstructing}
requires the interval to be bounded by $b<r_\text{max}$ where $r_\text{max}$ denotes a finite {\em execution horizon}.
In contrast, the typical event is defined in this paper to allow for results for infinite horizon.
%%DG]

%--------------------------------------------------------
\begin{lemma} \label{lemma: prob. of event G}
For all integers $1\le s < r$,
\begin{align}
    P(G[s,r]) > 1-{\color{black}5\eta^{-2}}e^{-\eta(r-s)}.
\end{align}
\end{lemma}
\begin{proof}
{\color{black}Due to the stationarity of $X$, $Y$ and $Z$ processes}, $P(E[s,r]) = P(E[1,r-s+1])$ for all $s, r$. {\color{black}Evidently the probability only} depends on the {\color{black}length of the interval} $r-s$.
\begin{align}
    P(G^c[s, r]) & = P(\cup_{0\le a < s, b\ge 0}E^c[s-a,r+b]) \\
    & = P(\cup_{0\le a < s, b\ge 0}E^c[1,r-s+a+b+1])\\
    & \le \sum_{0\le a < s, b\ge 0}P(E^c[1,r-s+a+b+1]) \\
    & = \sum_{k=0}^{\infty} \sum_{0\le a < s, b\ge 0:a+b=k}P(E^c[1,r-s+k+1])\\
    & < \sum_{k=0}^{\infty} (k+1)P(E^c[1,r-s+k+1])\\
    & < \sum_{k=0}^{\infty}(k+1)4e^{-\eta (r-s+k)}\\
    & = 4e^{-\eta (r-s)}\sum_{k=0}^{\infty} (k+1) e^{-\eta k}\\
    & = \frac{4}{(1-e^{-\eta})^2}e^{-\eta(r-s)}.
\end{align}
{\color{black}A}ccording to \eqref{equ: honest majority} and \eqref{def: gamma}, $\eta \le \frac{1}{6}\cdot \frac{1}{180} = \frac{1}{1080}$. {\color{black}The lemma is thus established using the fact that $1-e^{-x} \ge \sqrt{\frac{4}{5}}x$ for all $0\le x \le \frac{1}{1080}$.}
\end{proof}

%--------------------------------------------------------
\begin{lemma} \label{lemma: equal length}
{\color{black}A}ll honest  blockchains {\color{black}must have} identical length {\color{black}by} every round.
\end{lemma}
\begin{proof}
This is a simple consequence of the fact that all honest miners have seen the same blocks and every honest miner adopts {\color{black}the} longest blockchain at the end of every round.
\end{proof}

%-------------------------------------------------------

\begin{lemma} \label{lemma: unique block}
(Lemma 6 in \cite{garay2015bitcoin}) Suppose some blockchain's $k$th block $B$ is mined by an honest miner in a uniquely successful round. Then the $k$th block of every blockchain is either $B$ or an adversarial block.
\end{lemma}
\begin{proof}
Suppose the $k$th block of another blockchain is an honest block $B'\ne B$.
Let $r$ and $r'$ denote the rounds in which $B$ and $B'$ are mined, respectively.
Then we must have $r\ne r'$ by assumption that $B$ is mined in a uniquely successful round.
Since both $B$ and $B'$ are mined and adopted as the $k$th block by some honest miners, all other honest miners must have adopted a blockchain of length at least $k$ by round $r^*= \min\{r{\color{black}+1},r'{\color{black}+1}\}$.
Hence, all honest blocks mined after round $r^*$ will extend a blockchain longer than $k$. This contradicts the assumption that $B$ and $B'$ are both at position $k$ of some miner's blockchain.
Hence the proof of Lemma \ref{lemma: unique block}.
\end{proof}

%--------------------------------------------------------

\begin{lemma} \label{lemma: blockchain growth}
(Lemma 7 in \cite{garay2015bitcoin}) Let $1 \le s < r$ be integers. Suppose an honest blockchain is of length $l$ by round $s$. Then by round $r$, the length of every honest blockchain is at least $l + X[s,r]$.

\end{lemma}
\begin{figure}\centering
\includegraphics[clip, trim=0cm 3cm 20cm 0cm,width=0.4\columnwidth]{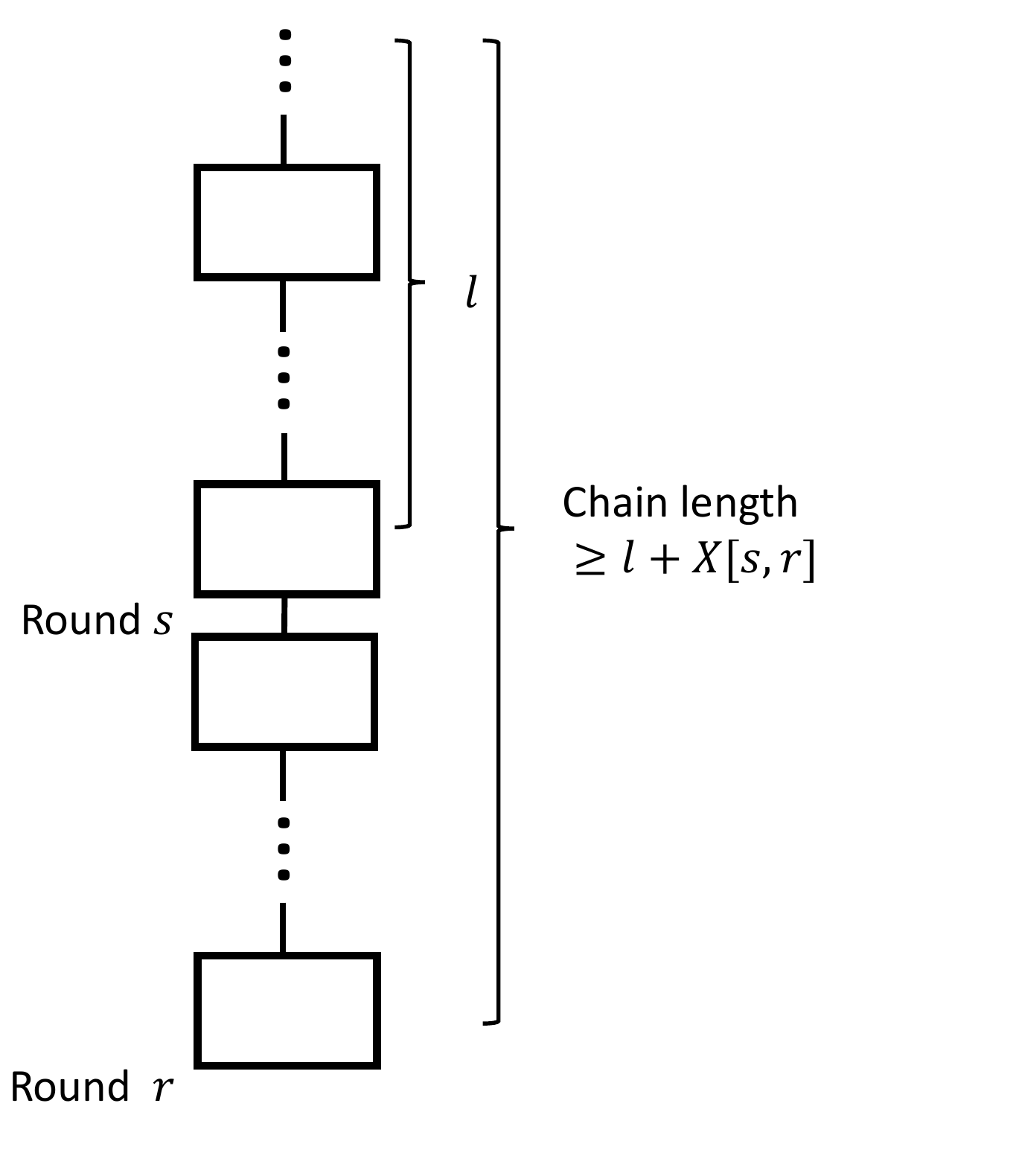}
\caption{Illustration for Lemma \ref{lemma: blockchain growth}.}
\end{figure}
\begin{proof}
By induction: Consider $r=s+1$. All honest miners' blockchains are of identical length $l$ by round $s$ according to Lemma \ref{lemma: equal length}.
If $X[s]=0$, then $X[s,s+1]=0$.
If $X[s]=1$, at least one honest block is broadcast to all miners during round $s$. Then by round $s+1$, each honest miner will adopt a blockchain of at least $l+1$ blocks. Thus Lemma \ref{lemma: blockchain growth} is established for the cases of $r = s+1$.

Assume by round $r_1$, each honest miner's blockchain length is at least $l+ X[s,r_1]$. If $X[r_1] = 0$, the claim holds trivially for round $r_1+1$. If $X[r_1]=1$, at least one honest miner will have a blockchain of length no shorter than $l + X[s, r_1] +1$ by round $r_1$. Then according to Lemma \ref{lemma: equal length}, each honest miner will adopt a blockchain of length at least $l+X[s,r_1+1]$ by round $r_1+1$. By induction on $r_1$, Lemma \ref{lemma: blockchain growth} holds.
\end{proof}

%------------------------------------------------------

\begin{lemma} \label{lemma: blockchain growth2 at least rounds}
(Blockchain growth lemma) For all integers $1 \le s < r$ and $k \ge 2q(r-s)$, under typical event $G[s, r]$,
every honest miner's $k$-deep block by round $r$ must be mined before round $s$.
\end{lemma}
\begin{proof}
The blockchain growth of an honest miner during rounds $\{s,\ldots,r-1 \}$ is upper bounded by $ X[s,r] + Z[s,r]$. Note that
\begin{align}
X[s,r] + Z[s,r] < & (1+\ep)q(r-s) + (1-\epq{2}{3})q(r-s) \label{6.01}\\
< & 2q(r-s) \\
\le & k,
\end{align}
where \eqref{6.01} is due to  \eqref{equ: X < (1+1/6del)qr} and \eqref{equ: Z < (1-2/3del)qr}. Thus, the $k$-deep block must be mined before round $s$.
\end{proof}
\begin{figure}\centering
\includegraphics[clip, trim=0cm 10cm 20cm 0cm,width=0.4\columnwidth]{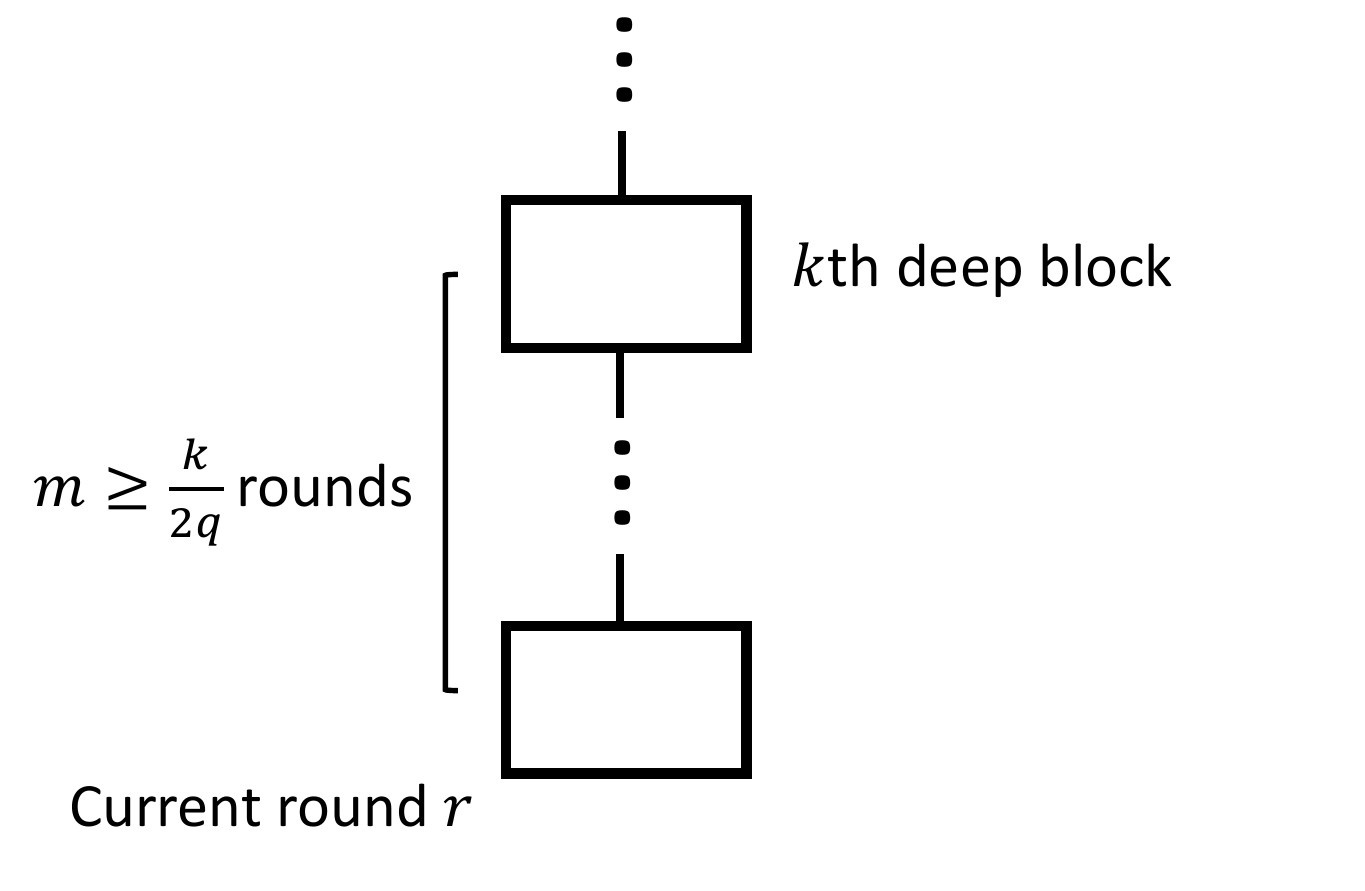}
\caption{Illustration for Lemma \ref{lemma: blockchain growth2 at least rounds}.}
\end{figure}

%----------------------------------------------------------

\begin{theorem} \label{thm: blockchain growth}
(Blockchain growth theorem)
Let $r, s, s_1$ be integers satisfying $1\le s_1\le s<r$. Then under typical event $G[s, r]$,  {\color{black}the length of} {\color{black}every honest blockchain} {\color{black} must increase by} at least $(1-\ep) q(r-s_1)$ during {\color{black}rounds $\{s_1, \ldots, r\}$}.
\end{theorem}
\begin{figure}\centering
\includegraphics[clip, trim=0cm 6cm 20cm 0cm,width=0.4\columnwidth]{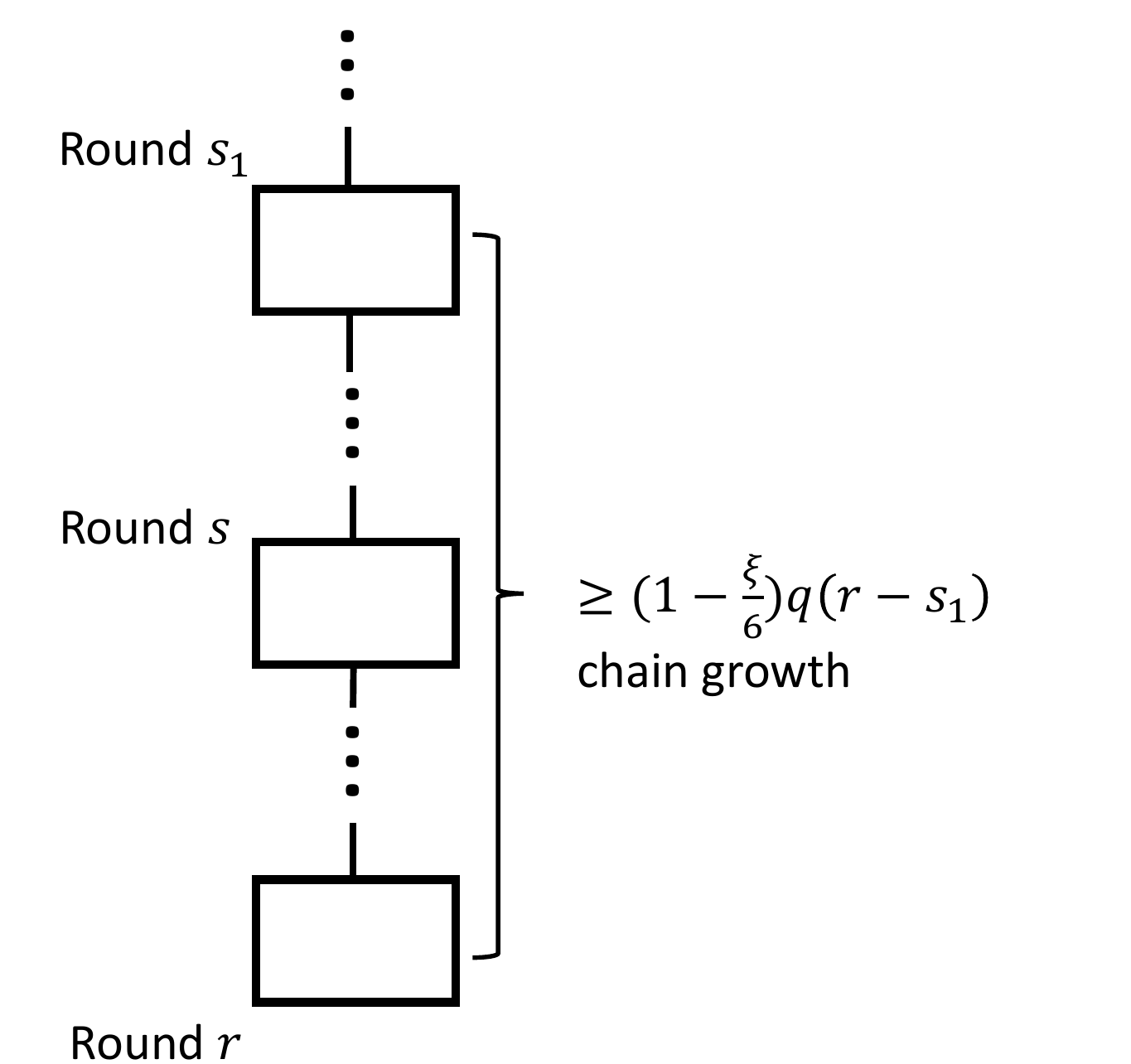}
\caption{Illustration for blockchain growth theorem.}
\end{figure}
\begin{proof}
Under $G[s,r]$,
\begin{align}
    X[s_1,r] & > (1-\ep)\mathbb{E}[X[s_1,r]] \\
    & =  (1-\ep)q (r-s_1) \label{3.11}
\end{align}
where \eqref{3.11} is due to \eqref{equ: X < (1+1/6del)qr}.
According to Lemma \ref{lemma: blockchain growth}, the blockchain growth  for any honest miner is at least $X[s_1,r]$ during $[s_1, r]$.
\end{proof}

Let $len(C)$ denote the length of a blockchain $C$.
%-----------------------------------------------------
\begin{theorem} \label{thm: blockchain quality}
(Blockchain quality theorem) Let $r,s,k$ be integers satisfying $1\le s < r$ and $k\ge 2q(r-s)$. Suppose an honest {\color{black}miner's} blockchain has more than $k$ blocks {\color{black}by round $r$}. Under event $G[s, r]$, by round $r$, at least $\epp{2}$ fraction of the last $k$ blocks of this miner's blockchain are honest.
\end{theorem}
\begin{proof}
The intuition is {\color{black}that} under typical event $G[s,r]$, an honest miner's blockchain grow by at least $X[s,r]$ {\color{black}according to Lemma \ref{lemma: blockchain growth}}.
{Meanwhile,} the number of adversarial blocks mined is {\color{black}upper bounded by \eqref{equ: Z < X2}}.
Thus, {\color{black}at} least $\epp{2}$ fraction of blocks must be honest even in the worst case that all adversarial blocks {\color{black}are included in the blockchain}.

To be precise, assume an honest miner adopts blockchain $C$ by round $r$. Denote $B_i$ as the $i$th block of blockchain $C$ ($C = B_0B_1\ldots B_{len(C)-1}$, where $B_0$ is the genesis block).
By assumption, $len(C)>k$.
Let $u=len(C)-k$. Then the last $k$ blocks of $C$ are $B_u\ldots B_{len(C)-1}$.
Let $B_{u'}$ be the last honest block before $B_{u}$.
That is to say, $u' = \max \{u'|u'\le u-1, B_{u'} \text{ is honest} \}$ ($u'$ is always well defined as $B_0$ is regarded as honest). Let $r^*$ be the round when $B_{u'}$ is mined. By Lemma \ref{lemma: blockchain growth2 at least rounds}, $r^*<s$.
Let $L=len(C)-u'-1$. Note that $L \ge k$. These definitions are illustrated in Figure \ref{fig:thmChainQuality}.
\begin{figure}\centering
\includegraphics[clip, trim=0cm 0cm 17cm 0cm,width=0.5\columnwidth]{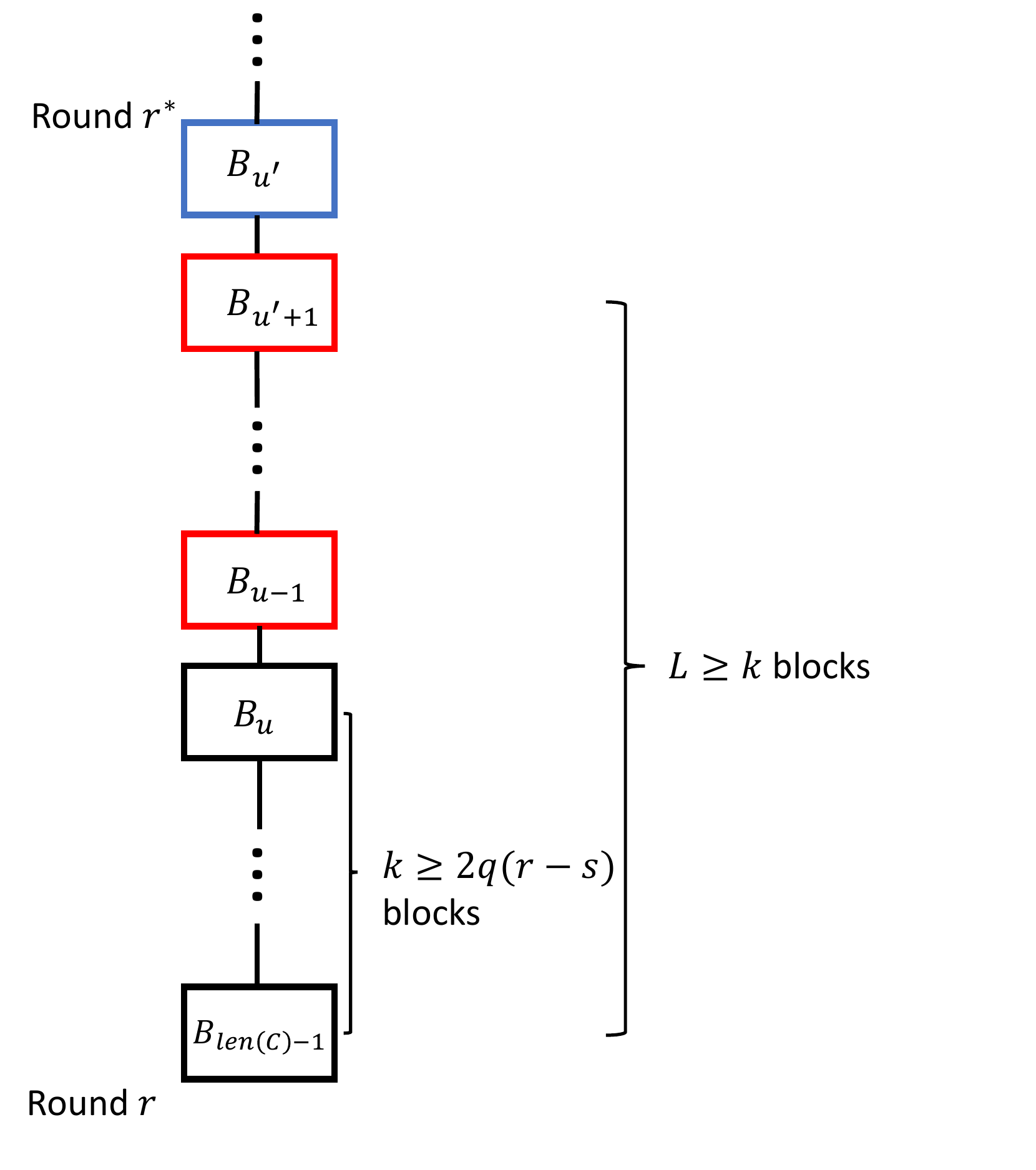}
\caption{Illustration to prove blockchain quality theorem\label{fig:thmChainQuality}.}
\end{figure}

Let $x$ be the number of honest blocks in $B_u\ldots B_{len(C)-1}$. To prove the theorem, it suffices to show $x > \epp{2} k$.  Since all blocks in $B_{u'+1}\ldots B_{u-1}$ are {\color{black}adversarial}, the number of honest blocks in  $B_{u'+1}\ldots {\color{black}B_{len(C)-1}}$ is also $x$. Thus, the number of adversarial blocks in $B_{u'+1}\ldots B_{len(C)-1}$ is $L - x$. Under $G[s,r]$, which implies that ${\color{black}E[r^*+1,r]}$ also occurs, we have
\begin{align}
L -x  & \le Z[r^*+1,r] \\
& < (1-\epp{2})X[r^*+1, r] \label{7.2}\\
& \le (1-\epp{2})L \label{7.3}\\
& \le L - \epp{2} k \label{7.4},
\end{align}
where \eqref{7.2} is due to \eqref{equ: Z < X2}, \eqref{7.3} is due to Lemma {\color{black}\ref{lemma: blockchain growth}}, and \eqref{7.4} is due to $L\ge k$. From \eqref{7.4}, $x>\epp{2}k$ is derived.
\end{proof}
%-----------------------------------------------------

Let $C^{\lceil k}$ denote the $k$-deep prefix of blockchain $C$. If $len(C)\le k$, let $C^{\lceil k}$ be the genesis block.

%--------------------------------------------------------------------------
\begin{definition}\label{def: permanent seq}
Let $G$ be an event and $r$ be a positive integer. A block or a sequence (of blocks)  is said to be {\em permanent after round $r$ under $G$} if, under event $G$, the block or sequence remains in all honest blockchains starting from round $r$.
\end{definition}

%------------
\begin{definition}\label{def: eps permanent seq}
Let $r$ be a positive integer. A block or a sequence (of blocks) is said to be {\em $\epsilon$-permanent after round $r$} if, there exists an event $G$ with $P(G)>1-\epsilon$ such that the block or sequence is permanent after round $r$ under $G$.
\end{definition}

\begin{lemma} \label{lemma: permanent ever after}
If a bock or {\color{black}a} sequence is $\epsilon$-permanent after round $r$, then it is also $\epsilon$-permanent after round $s$ for every $s>r$.
\end{lemma}

%------------------------------------------------

%-----------------------------------------------
\begin{theorem}\label{thm: common prefix}
(Common prefix theorem) Let $r, s ,k$ be integers satisfying $1\le s < r$ and $k\ge 2q(r-s)$. If by round $r$ an honest blockchain has a $k${\color{black}-deep} prefix, then  the prefix is permanent {\color{black}after round $r$ under $G[s,r]$ }.
\end{theorem}
\begin{proof}
The intuition is based on Lemma \ref{lemma: unique block}: Once a block is mined in a uniquely successful round, a different block on any other {\color{black}blockchain} at the same position must be adversarial. If some adversarial miners wish to fork the blockchain, they must generate at least one adversarial block during every uniquely successful round after the common prefix. This can not be true because according to \eqref{equ: Z < Y}, the number of uniquely successful rounds must be greater than the number of adversarial blocks under the typical event.

To be precise, we prove the desired result by contradiction. Suppose blockchain $C_1$ whose length is great than $k$ is adopted by an honest miner $P_1$ by round $r$.  Contrary to the claim, assume $r_2 > r$ is the smallest round by which an honest miner $P_2$ adopts a blockchain $C_2$ such that $C_1^{\lceil k} \npreceq C_2$.
Let $C'_2$ be the blockchain $P_2$ adopted by round $r_2-1$. Note that $C_1^{\lceil k}\preceq C'_2$.

Assume the last honest block on the common prefix of $C'_2$ and $C_2$ is {\color{black}mined during} round $r^*$.  If $r^*>0$,  this common block of $C'_2$ and $C_2$ must  be more than $k$  {\color{black}deep in $C_1$ by round $r$.}  According to Lemma \ref{lemma: blockchain growth2 at least rounds}, we have {\color{black}$r^*< s$}, so that
\begin{align}\label{8.01}
[s,r] \subset [r^*+1, r_2-1].
\end{align}
On the other hand, if $r^*=0$, the last common block is the genesis block. Since $s\ge 1$, \eqref{8.01} also holds. An illustration of these chains and parameters is given in Figure \ref{fig:thmCommonPrefix}.

\begin{figure}\centering
\includegraphics[clip, trim=0cm 0cm 17.5cm 0cm,width=0.5\columnwidth]{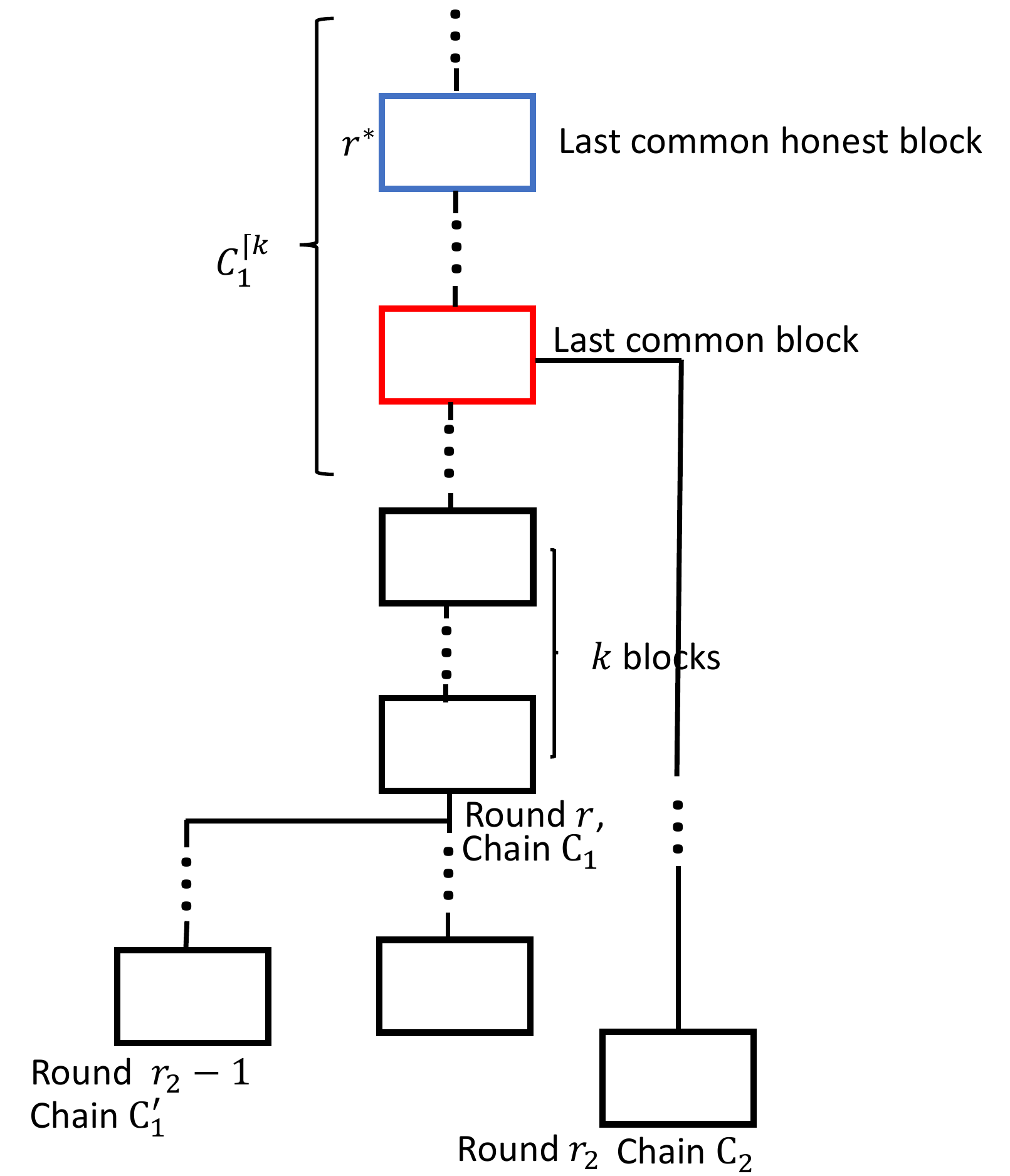}
\caption{Illustration to prove common prefix theorem\label{fig:thmCommonPrefix}.}
\end{figure}

By assumption $G[s,r]$, {\color{black}$E_2[s,r]$ also occurs, so that $Y[s,r]>0$} according to \eqref{equ: E2}. {\color{black}Hence,} there must be at least one uniquely successful round $u\in \{r^*+1, \ldots, r_2-2\}$. According to Lemma \ref{lemma: equal length}, all honest miners have the same {\color{black}chain} length.  Let $l_u$ denote one plus the length of {\color{black}the} honest miners' blockchains by round $u$. Suppose honest miner $P$ mine{\color{black}s} $B_u$ during round $u$.
According to Lemma \ref{lemma: unique block}, the $l_u$th block of every blockchain is either $B_u$ or an adversarial block.
Because $C_2$ and $C'_2$ are adopted by an honest miner after round $r_2-2$, they must be no shorter than $\max \{l_u: u\; \text{is a uniquely successful round in}\; \{r^*+1,\ldots,r_2-2\} \}$.

For every uniquely successful round {\color{black}$u$} in $\{r^*+1,\ldots,r_2-2\}$, if the $l_u$th blocks of $C_2$ and $C'_2$ are different, then at least one of them must be adversarial according to Lemma \ref{lemma: unique block}. On the other hand, if the $l_u$th block of $C_2$ and $C'_2$ are identical, the block must be in their common prefix, which must be adversarial by definition of $r^*$. Thus, at least one adversarial block is mined during each uniquely successful round, {\color{black}so that} $Z[r^* +1, r_2-1] \ge Y[r^*+1, r_2-1]$. However, since $[s,r]\subset [r^*+1, r_2-1]$,  $E[r^*+1,r_2-1]$ occurs under $G[s,r]$, {\color{black}so that} $Z[r^* +1, r_2-1] < Y[r^*+1, r_2-1]$ according to \eqref{equ: Z < Y}. {\color{black}Contradiction arises. Hence the proof of the theorem.}
\end{proof}

%--------------------------------------------------------------------------

\section{The Prism backbone protocol} \label{sec: prism}
The Prism protocol is {\color{black}invented and fully described} in \cite{bagaria2018deconstructing}. Here we describe the Prism backbone with just enough detail{\color{black}s} to facilitate its analysis.
We assume {\color{black}$m+1$ genesis blocks are generated for the same number of} blockchains during round $0$ by honest miners. {\color{black}Blockchain $0$ is referred to as
the proposer blockchain. The remaining blockchains are voter blockchains.}
 A block is mined before knowing which blockchain it will be part of. Sortition relies on the range the nonce's hash lands in: If a miner find a nonce whose hash is within $[j\alpha, j\alpha + \alpha)$ for $j= 0,1,\ldots,m$, the mined block belongs to blockchain $i$. Mining difficulty can be adjusted by changing parameter $\alpha$.  This sortition scheme ensures the mining power of both honest and adversarial miners are evenly distributed across different voting blockchains and the proposer blockchain.

To certify its level, a new honest voter block for blockchain $j$ ($j=1,2,\ldots,m$) points to blockchain $j$'s maximum-level block by a {\em parent link} ({\color{black}ties are broken by} predefined rules).
To certify its level, an honest new proposer block includes the hash of a maximum-level block in the proposer blockchain and point to it by a {\em reference link}. In addition, an honest new proposer includes one reference link to every {\color{black}existing} block in both proposer and voter blockchains {\color{black}that} has not been pointed to by other reference links.

{\color{black}Following the bitcoin protocol, an honest} {\color{black}miner} decides {\color{black}each} main voter blockchain by the longest  blockchain  rule. The {\color{black}miner determines the} its main blockchain by votes from {\color{black}the main} voter blockchains. Let $B$ be an honest block on {\color{black}a} voter blockchain $j$. {\color{black}By} $B$'s ancestors {\color{black}we} mean all blocks on $B$'s path to the blockchain genesis block {\color{black}following parent links}. By saying $B$ votes for a level $l$, we mean $B$ chooses one proposer block among all proposer blocks at level $l$ according to a predefined rule, and points to its choice with a reference link. {\color{black}An honest} voter block vote{\color{black}s} for all levels which have not been voted by its ancestors.

A voter blockchain is allowed to vote only once for each level (more {\color{black}votes} from the same voter blockchain {\color{black}are} discarded). That is to say, proposer blocks {\color{black}on} the same level receive $m$ votes in total. At each level, the proposer block with most votes is elected as a leader block{\color{black}, with ties broken by a predefined rule}. The sequence of leader blocks over all levels is called the leader sequence.

{\color{black}A miner generates} its final ledger {\color{black}based on its} leader sequence. Given a leader sequence $B_0B_1\ldots B_l$, each leader block $B_i$ defines an epoch. Added to the ledger are the blocks which are pointed to by $B_i$, as well as other blocks reachable from $B_i$ but have not been included in previous epochs. The list of blocks are sorted topologically, {\color{black}with} ties broken by their contents. Since the blocks referenced are mined independently, there can be double spends or redundant transactions. An end user can create a valid ledger by keeping only the first transaction among double spends or redundant transactions.

For $j =0,1,\ldots,m$ and  $r = 1,2,\ldots$ , let $H_j[r]$ denote the {\color{black}total} number of honest blocks mined during round $r$ for blockchain $j$. {\color{black}Following} the definitions in {\color{black}Section \ref{sec: bitcoin backbone}}, for $j = 0,1,\ldots, m$ and $r = 1,2,\ldots$, we also define
\begin{align}
    X_j[r] =
    \begin{cases}
    1, \; \; & \text{if}\; H_j[r] \ge 1\\
    0, & \text{otherwise},
    \end{cases}
\end{align}
\begin{align}
    Y_j[r] =
    \begin{cases}
    1, \; \; & \text{if}\; H_j[r] = 1\\
    0, & \text{otherwise},
    \end{cases}
\end{align}
{\color{black}and let} $Z_j[r]$  be the {\color{black}total} number adversarial blocks {\color{black}mined} for blockchain $j$ during round $r$.

%----------------------------------------------------

\begin{definition} \label{def: Ej}
For all integers $1\le s<r$ and  $0\le  j \le m$, define event
\begin{align}
    E_j[s,r] := E_{1,j}[s,r] \cap E_{2,j}[s,r] \cap E_{3,j}[s,r]
\end{align}
where
\begin{align}
   E_{1,j}[s,r] &:= \left\{ (1-\ep)\mathbb{E}[X_j[s,r]] < X_j[s,r] < (1+\ep)\mathbb{E}[X_j[s,r]]\right\}\\
   E_{2,j}[s,r] &:=\left\{ (1-\ep)\mathbb{E}[Y_j[s,r]] < Y_j[s,r]\right\}\\
   E_{3,j}[s,r] &:= \left\{Z_j[s,r] < \mathbb{E}[Z_j[s,r]] + \ep \mathbb{E}[X_j[s,r]]\right\}.
\end{align}
\end{definition}

%----------------------------------------------------------

{\color{black}We note} that for integers $0\le j \le m$ and $r\ge 1$,  $H_j[r], X_j[r], Y_j[r],$ and $Z_j[r]$ {\color{black} here are identically distributed as} $H[r], X[r], Y[r],$ and $Z[r]$ {\color{black} defined in Section \ref{sec: bitcoin backbone}}. Also, for $1\le s < r$, $E_j[s,r]$ is defined {\color{black} in} the same manner as $E[s,r]$. Thus,  the proposer blockchain and all voter blockchains satisfy similar properties as in Lemma \ref{lemma: good properties}:

%---------------------------------------------------------

\begin{lemma} \label{lemma: good properties for j}
(Typical properties lemma for proposer and voter blockchain) For all integers $1\le s < r$ and $0\le j \le m$, under event $E_j[s,r]$, the following holds{\color{black}:}
\begin{align}\label{equ: X_j < (1+1/6del)qr}
(1-\ep)q(r-s) < X_j[s,r]  < (1+\ep)q(r-s)
\end{align}
\begin{align} \label{equ: Y_j > (1-1/del)qr}
Y_j[s,r]  >(1-\epp{3})q(r-s)
\end{align}
\begin{align}\label{equ: Z_j < (1-2/3del)qr}
Z_j[s,r] < (1-\epq{2}{3})q(r-s)
\end{align}
\begin{align} \label{equ: Z_j < X_j2}
 Z_j[s,r]   < (1-\epp{2})X_j[s,r]
\end{align}
\begin{align} \label{equ: Z_j < Y_j}
Z_j[s,r] < Y_j[s,r].
\end{align}
\end{lemma}
\begin{proof}
{\color{black}For $j = 0,1,\ldots,m$, the lemma admits essentially the same proof as that for Lemma \ref{lemma: good properties}.}
\end{proof}
%---------------------------------------------------------
\begin{definition} \label{lemma: event Gj}
For all integers $1\le s < r$ and $0\le j\le m$, {\color{black}define blockchain $j$'s typical event with respect to $[s,r]$ as}
\begin{align}
    G_j[s, r] := \cap_{0\le a < s, b\ge 0}E_j[s-a,r+b].
\end{align}
\end{definition}

%--------------------------------------------------------
\begin{lemma} \label{lemma: prob of typical event j}
For all integers $1\le s<r$ and $0\le j\le m$,
\begin{align}
    P(G_j[s, r]) > 1-{\color{black}5\eta^{-2}}e^{-\eta(r-s)}
\end{align}
where $\eta$ is defined in \eqref{def: gamma}.
\end{lemma}
\begin{proof}
{\color{black}For $j = 0,1,\ldots,m$, the lemma admits essentially the same proof as that for Lemma \ref{lemma: prob. of event G}.}
\end{proof}
%------------------------------------------------------------

Since the proposer blockchain and all voter blockchains grow in the same manner as how a bitcoin blockchain grows, the blockchain growth lemma and blockchain growth theorem remain valid{\color{black}:}
{\color{black}
\begin{lemma} \label{lemma: blockchain j growth}
Let $1 \le s < r$ and $0\le j \le m$ be integers. Suppose an honest voter blockchain $j$ is of length $l$ by round $s$. Then by round $r$, the length of every honest voter blockchain
$j$ is at least $l + X_j[s,r]$.
\end{lemma}
\begin{proof}
For $j = 0,1,\ldots,m$, the lemma admits essentially the same proof as that for Lemma \ref{lemma: blockchain growth}.
\end{proof}
}

\begin{lemma} \label{lemma: blockchain j growth2 at least rounds}
(Blockchain growth lemma for voter and proposer blockchain) For all integers $1 \le s < r$, $k \ge 2q(r-s)$ and $0\le j \le m$, under typical event $G_j[s, r]$, every honest miner's $k$-deep {\color{black}block} of blockchain $j$ by round $r$ must be mined before round $s$.
\end{lemma}
\begin{proof}
For $j = 0,1,\ldots,m$, the lemma admits essentially the same proof as that for Lemma \ref{lemma: blockchain growth2 at least rounds}.
\end{proof}
%----------------------------------------------------------

\begin{theorem} \label{thm: blockchain j growth}
(Blockchain growth theorem for voter and proposer blockchain)
Let $r, s, s_1$ be integers satisfying $1\le s_1\le s<r$. Let $j$ be an integer satisfying $0\le j\le m$. Then under typical event $G_j[s, r]$,   the length of every honest miner's blockchain $j$ {\color{black}must grow by at least $(1-\ep) q(r-s_1)$} during rounds $\{s_1, \ldots, r\}$.
\end{theorem}
\begin{proof}
{\color{black}For $j = 0,1,\ldots,m$, the lemma admits essentially the same proof as that for {\color{black}Theorem} \ref{thm: blockchain growth}.}
\end{proof}
%----------------------------------------------------------

Since the {\color{black}protocol} {\color{black}for} voter blockchains is identical to that of bitcoin, the blockchain quality theorem and  {\color{black}the} common prefix theorem hold for all voter blockchains.

%----------------------------------------------------------
\begin{theorem} \label{thm: blockchain j quality}
(Blockchain quality theorem for voter blockchain) Let $r,s,k, j$ be integers satisfying $1\le s < r$, $k\ge 2q(r-s)$ and $1\le j \le m$.  Suppose an honest {\color{black}blockchain has} more than $k$ blocks {\color{black}by round $r$}. Under event $G_j[s, r]$, at least $\epp{2}$ fraction of the last $k$ blocks of this blockchain $j$ are honest.
\end{theorem}
\begin{proof}
{\color{black}For $j = 0,1,\ldots,m$, the lemma admits essentially the same proof as that for Theorem \ref{thm: blockchain quality}.}
\end{proof}
%-----------------------------------------------------

\begin{theorem}\label{thm: common prefix j}
(Common prefix theorem for voter blockchain) Let $r, s ,k, j$ be integers satisfying $1\le s < r$, $k\ge 2q(r-s)$ and $1\le j \le m$. If by round  $r$ {\color{black}an honest} voter blockchain $j$ has a $k$ prefix, then  the prefix is {\color{black}permanent after round $r$ under $G_j[s,r]$.}

\end{theorem}
\begin{proof}
{\color{black}For $j = 0,1,\ldots,m$, the lemma admits essentially the same proof as that for Theorem \ref{thm: common prefix}.}
\end{proof}

%-----------------------------------------------------------

{\color{black}Since} the leader sequence of {\color{black}the} proposer blockchain is decided by votes instead of the longest  blockchain  rule, the blockchain quality theorem and common prefix theorem do not immediately extend to the leader sequence of {\color{black}the} proposer blockchain.

%-------------------------------------------------------

\begin{definition} \label{d:Rlek}
    We define $R_l$ as the round in which the first proposer block on level $l$ is mined.
    We also define
    \begin{align} \label{def: epsilon}
    \epsilon_k=6m\eta^{-2}e^{-\eta \frac{k}{2q}},  {\color{black}\quad k = 1,2,\ldots }
    \end{align}
    where $\eta$ is given by \eqref{def: gamma}.
\end{definition}

%--------------------------------------------------------
\begin{definition}
We let $\textbf{LedSeq}_l(r)$ denote the proposer blockchain's leader sequence up to level $l$ by round $r$.
\end{definition}

\begin{lemma} \label{lemma: one honest block fix all}
Consider a given level $l$. {\color{black}L}et  $k$ be a positive integer.
If by some round  $r > \max\left\{\frac{k}{2q}, R_l+1\right\}$, every voter blockchain  contains at least one honest block mined after round $R_l$ which is at least $k$-deep, then {\color{black}$\textbf{LedSeq}_l(r)$} is  $\epsilon_k$-permanent after round $r$.
\end{lemma}

\begin{proof}
{\color{black}Let
\begin{align} \label{9.-25}
    s=r-\left\lfloor \frac{k}{2q} \right\rfloor,
\end{align}which}
must be a positive integer because $2qr > k$. Define
\begin{align}
    G =\cap_{j=1,2,\ldots,m}G_j[s, r].
\end{align}

For $j=1,\dots,m$, let $B_j$ denote an honest block  on {\color{black}an honest} voter blockchain $j$ {\color{black}which} is mined after round $R_l$ and is at least $k$-deep by round $r$. According to Theorem \ref{thm: common prefix j}, $B_j$ and its {\color{black}ancestors} are permanent after round $r$ under $G_j[s, r]$  . {\color{black}Hence,} $B_1,\ldots,B_m$ and all their {\color{black}ancestors} must be permanent after round $r$ under $G$. Thus, all voter blockchains' voting are permanent. Since $B_1,\ldots,B_m$ are honest, they would have voted for all levels up to level $l$ of the proposer blockchain by the voting rule. Hence, the leader block sequence up to level $l$ is permanent after round $r$ under $G$.
Note that
\begin{align}
P(G) & = 1-P(\cup_{j=1,2,\ldots,m}G^c_j[s, r]) \\
& \ge 1 - \sum_{j=1}^m P(G^c_j[s, r])\label{9.-1}\\
& = 1-mP(G^c_1[s,r]) \label{9.-05}\\
& > 1-  {\color{black}5m\eta^{-2}}e^{-\eta(r-s)},  \label{9.0}
\end{align}
where \eqref{9.-1} is due to {\color{black}the} union bound, \eqref{9.-05} is due to symmetry of all {\color{black}voter} blockchains, \eqref{9.0} is due to Lemma \ref{lemma: prob of typical event j}.
By \eqref{9.-25}, {\color{black}we have} $2q(r-s+1)>k$, {\color{black}so that \eqref{9.0} becomes}
\begin{align}
P(G) & > 1-  {\color{black}5m\eta^{-2}}e^{-\eta  \frac{k}{q}+\eta} \\
&{\color{black} > 1-  {\color{black}6m\eta^{-2}}e^{-\eta \frac{k}{2q}}} \label{9.15}\\
& = 1-\epsilon_k
\end{align}
where
\eqref{9.15} is due to $e^{\eta}< \frac{6}{5}$.
Thus, the leader block sequence up to level $l$ is $\epsilon_k$-permanent after round $r$.
\end{proof}

%--------------------------------------------------------
\begin{lemma} \label{lemma: get honest block}
{\color{black}If positive integers $R$, $r$, and $k$ satisfy}
\begin{align} \label{def: r 8.1}
r \ge  \frac{2(k+1)}{(1-\ep)\xi q}+1,
\end{align}
{\color{black}then right before round $R+r$,} with probability at least $1-\epsilon_k$, all honest voter blockchains have an honest block mined after round $R$ which is at least $k$ deep.
\end{lemma}
\begin{proof}
Let
\begin{align}
    \ell = \left\lceil  \frac{2k}{\xi} \right\rceil. \label{8.111}
\end{align}
Let
\begin{align}
    s_1 = \left\lfloor \frac{k}{q\xi}\right\rfloor.  \label{7.95}
\end{align}
Then
\begin{align}
    \ell&\ge \frac{2k}{\xi} \\
    & \ge 2q \left\lfloor \frac{k}{q\xi}\right\rfloor \\
    & = 2qs_1. \label{8.105}
\end{align}

According to the Theorem \ref{thm: blockchain j growth}, under event $G_j[R, R+r]$, an honest {\color{black}voter} blockchain $j$'s growth during $\{R, R+1,\ldots, R+r-1\}$ is at least
\begin{align}
    (1-\ep)qr & \ge \frac{2k}{\xi}+1 \label{8.112} \\
    & > \ell, \label{8.113}
\end{align}
where \eqref{8.112} is due to \eqref{def: r 8.1} and \eqref{8.113} is due to \eqref{8.111}.

According to Theorem \ref{thm: blockchain j quality} {\color{black}and \eqref{8.105}},  under event $G_j[R + r - s_1, R + r]$, at least $\epp{2}$ fraction of the last $\ell$ blocks of {\color{black}this} voter blockchain $j$ are honest. {\color{black}Because $\frac{\xi}{2}\ell \ge k$}, the earliest of these  honest blocks must be at least $k$ deep.

By \eqref{def: r 8.1} and \eqref{7.95}, it is easy to see that $s_1\le r$. Hence $ G_j[R+r-s_1, R+r]\subset G_j[R, R+r]$. Define
\begin{align}
    G = \cap_{j=1,2,\ldots,m}G_j[R+r-s_1, R+r] .
\end{align}
Under event $G$, by round $R+r$, every {\color{black}honest} voter blockchain has an honest block mined after round $R$ which is at least $k$ deep. {\color{black}The probability of the typical event can be lower bounded:}
\begin{align}
P(G) & = P(\cap_{j=1,2,\ldots,m}G_j[R+r-s_1, R+r]) \\
& = 1-P(\cup_{j=1,2,\ldots,m}G^c_j[R+r-s_1, R+r])\label{8.0}\\
& \ge 1-mP(G^c_1[R+r-s_1, R+r]) \label{8.1}\\
&> 1- {\color{black}5}m\eta^{-2}e^{-\eta s_1} \label{8.2} \\
& > 1- 6m\eta^{-2}e^{-\eta s_1},
\end{align}
where \eqref{8.1} is due to the union bound and symmetry of all voter blockchains and \eqref{8.2} is due to  Lemma \ref{lemma: prob of typical event j}. Moreover,
\begin{align}
    qs_1 & > q\left( \frac{k}{q\xi}-1\right) \label{8.21}\\
    & = \frac{k}{\xi}-q\\
    &{\color{black} > k - \frac{1}{6}} \\
    & > \frac{k}{2},
\end{align}
{\color{black}where \eqref{8.21} is due to \eqref{7.95}.}
Therefore,
\begin{align}
P(G) &  > 1- {\color{black}6m\eta^{-2}}e^{-\eta \frac{k}{2q}} \\
& = 1- \epsilon_k,
\end{align}

In summary, by round $R+r$, with probability at least $1-\epsilon_k$, {\color{black}all honest} voter blockchains have an honest block mined after round $R$ which is at least $k$ deep.
\end{proof}

%-------------------------------------------------------------
\begin{theorem} \label{thm: permanent leader block}
Fix {\color{black}$\epsilon \in (0,1)$}. Let $R_l$ be {\color{black}the round during which the first proposer block on level $l$ is mined}. {\color{black}For every integer}
\begin{align} \label{def: r thm permanent}
    r \ge  \frac{5}{(1-\ep) \xi \eta}\log {\color{black}\frac{12m\eta^{-2}}{\epsilon}},
\end{align}
{\color{black}the leader sequence up to level $l$} is $\epsilon-$permanent after round $R_l + r$.
\end{theorem}

\begin{proof}
 Let
\begin{align}
k = \left\lceil \frac{2q}{\eta}\log \frac{12m\eta^{-2}}{\epsilon} \right\rceil, \label{11.01}
\end{align}
and
\begin{align}
    s = \left\lceil \frac{2(k+1)}{(1-\ep)\xi q}+1 \label{11.02} \right\rceil.
\end{align}
Let $\epsilon_k$ be as defined as in \eqref{def: epsilon}.

According to Lemma \ref{lemma: get honest block} and \eqref{11.02}, by round $R_l + s$, all honest voter blockchains have an honest block which is mined after $R_l$ and is at least $k$ deep with probability at least {\color{black}$1-\epsilon_k$}. {\color{black}Under this event}, according  to Lemma \ref{lemma: one honest block fix all} ({\color{black}evidently, $R_l+s > \frac{k}{2q}$}),  the leader sequence up to level $l$ is $\epsilon_k$-permanent after round $R_l+s$. Therefore, {\color{black}the leader sequence up to level $l$} is $2\epsilon_k$-permanent after round $R_l+s$. Note that
\begin{align}
\epsilon_k & = {\color{black}6m\eta^{-2}}e^{-\eta \frac{k}{2q}} \\
& \le {\color{black}6m\eta^{-2}}e^{-\log {\color{black}\frac{12m\eta^{-2}}{\epsilon}}} \label{11.11}\\
& =  \frac{\epsilon}{2}.
\end{align}
%Thus, after round $R_l + s$, with probability at least $1-\frac{\epsilon}{2}$, every voter blockchain has an honest block which is mined after $R_l$ and is at least $k$ deep. If every voter blockchain has an honest block which is mined after $R_l$ and is at least $k$ deep, the leader sequence up to level $l$ is $\frac{\epsilon}{2}$-permanent.
{\color{black}the leader sequence up to level $l$} is $\epsilon$-permanent after round $R_l+s$.

From \eqref{11.01}, it is easy to verify that $k>10$. As a consequence, we have
\begin{align}
 s  & <  \frac{2(k+1)}{(1-\ep)\xi q}+2 \label{11.05}\\
 & = \frac{2k+2+2(1-\ep)\xi q}{(1-\ep)\xi q} \\
 & < \frac{\frac{5}{2}(k-1)}{(1-\ep)\xi q} \label{11.051}\\
  & < \frac{5}{(1-\ep) \xi \eta}\log {\color{black}\frac{12m\eta^{-2}}{\epsilon}} \label{11.065}\\
 & \le r, \label{11.066}
\end{align}
where \eqref{11.05} is due to \eqref{11.01},
{\color{black}\eqref{11.051} is due to $k>10$, \eqref{11.065} is due to \eqref{11.01},
and \eqref{11.066} is by \eqref{def: r thm permanent}}.

Since $r>s$, the leader sequence up to level $l$ is $\epsilon$-permanent after round $R_l+r$ by Lemma \ref{lemma: permanent ever after}.
\end{proof}

%-----------------------------------------------------------

\begin{theorem} \label{thm: blockchain quality for proposer block}
(Blockchain quality theorem for proposer blockchain)
Let $r,s,k$ be integers satisfying $1 \le s < r$ and $k\ge 2q(r-s)$. Suppose an honest proposer blockchain has more than $k$ leader blocks {\color{black}by round $r$}. Under event $G_0[s, r]$, by round $r$, at least $\epp{2}$ fraction of the last $k$ leader blocks {\color{black}of the {\color{black}proposer} blockchain} are honest.
\end{theorem}
\begin{proof}
{\color{black}Let $l$ denote the highest level} of the proposer blockchain by round $r$. {\color{black}Evidently $l>k$.}
Let $l^*$ be the {\color{black}highest level} before $l-k+1$ {\color{black}on which} the first proposer block is honest. {\color{black}$l^*$ may be as high as $l-k$ and as low as $0$, which corresponds to the genesis block}.
Let $r^*$ be the round when the first block on level $l^*$ is mined. If this block is the genesis block, {\color{black}then} $r^* = 0$. If $r^*>0$, since blocks on level $l^*$ are more than $k$ blocks away from the last level by round $r$, {\color{black}we have} $r^* < s$ according to Lemma \ref{lemma: blockchain j growth2 at least rounds}. In any cases, we have $[s,r]\subset [r^*+1,r]$. {\color{black}An illustration of the said proposer blocks and blockchain is given in} Figure \ref{fig:thmChainQualityProposer}.

\begin{figure}\centering
\includegraphics[clip, trim=0cm 2cm 10cm 0cm,width=0.7\columnwidth]{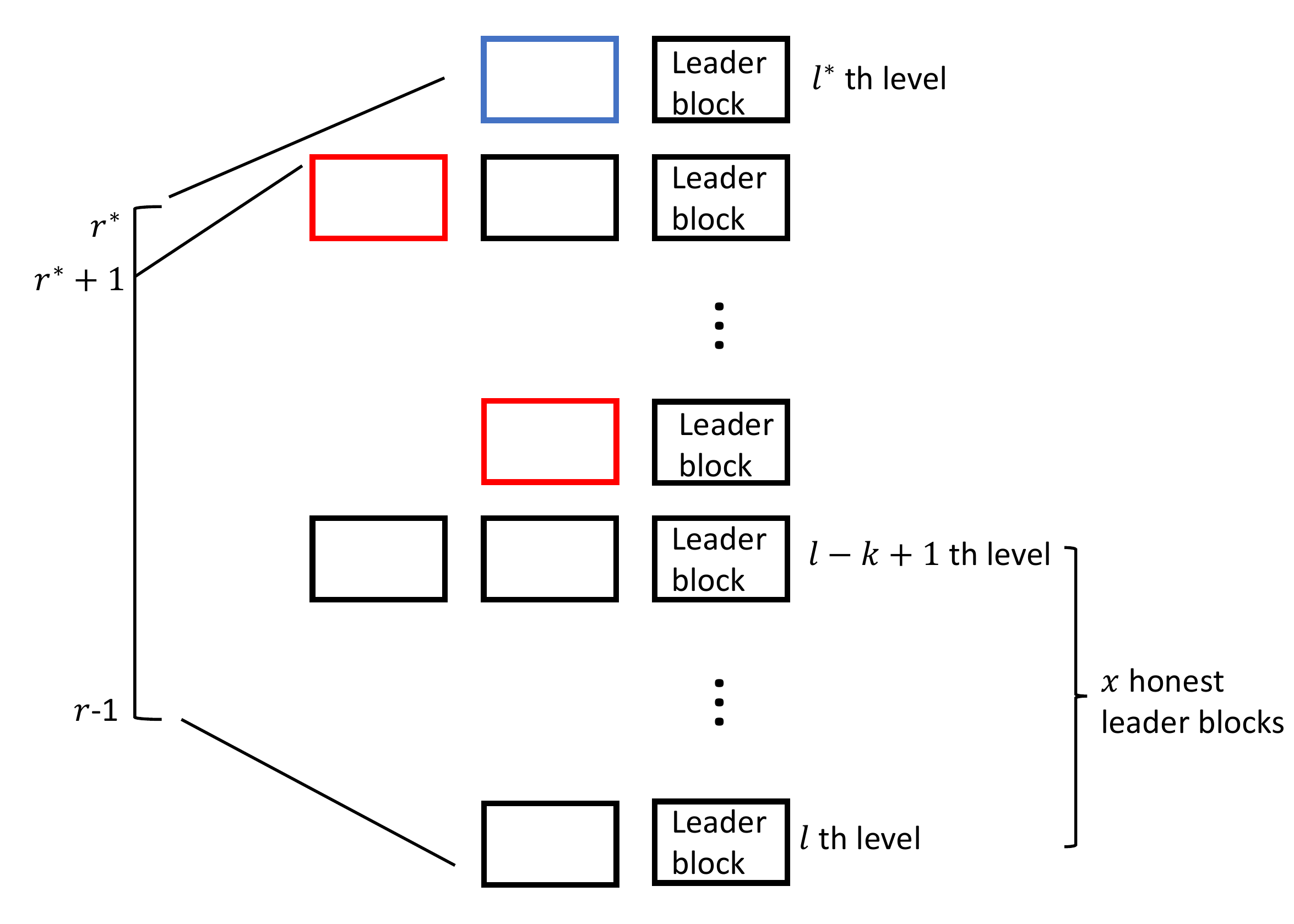}
\caption{Illustration to prove blockchain quality theorem for proposer blockchain\label{fig:thmChainQualityProposer}.}
\end{figure}

{\color{black}Since the first proposer block on every level within $\{l^*+1, \dots, l-k\}$ is adversarial,}
from level $l^*+1$ to level $l$, there must be at least one adversarial block on every level except (possibly) on the levels between $l-k+1$ and $l$ where the leading block is honest. Let $x$ be the number of honest leader blocks on levels $\{l-k+1,\ldots, l\}$. Then during rounds $\{r^*+1,\ldots, r-1\}$, the total number of adversarial {\color{black}proposer} blocks is no fewer than $l-l^*-x$, i.e.,
\begin{align}
Z_0[r^*+1, r] & \ge l-l^* - x. \label{10.-2}
\end{align}
Under $G_0[s,r]$, $E_0[r^*+1, r]$ occurs. Thus,
\begin{align}
x & \ge l - l^* - Z_0[r^*+1, r]\label{10.-1}\\
& > l-l^* - (1-\epp{2})X_0[r^*+1, r] \label{10.0}\\
& \ge \epp{2} (l-l^*) \label{10.3}\\
& \ge \epp{2} k, \label{10.4}
\end{align}
where \eqref{10.0} is due to \eqref{equ: Z_j < X_j2},  \eqref{10.3} is due to Lemma {\color{black}\ref{lemma: blockchain j growth}}, and {\color{black}\eqref{10.4}} is due to $l-l^* \ge k$.
To sum up,  {\color{black}we have $x>\epp{2} k$ and the proof is complete}.
\end{proof}

\begin{definition}
A  transaction $tx$ is  {\em honest} if it has been broadcast, and {\color{black}no other transaction spending from the same unspent output has been broadcast.}
\end{definition}
Note that the {\color{black}notion of honesty} is applicable only to transactions which have been broadcast.

\begin{definition}
A transaction is {\color{black}said to be} {\em $\epsilon$-permanent after round $r$} if{\color{black},} with probability at least $1-\epsilon$, it remains on the final ledger of every honest miner after round $r$.
\end{definition}

\begin{lemma} \label{lemma: honest leader block include blocks}
Suppose {\color{black}right before round $r$}, the leader block on level $l$ is honest. {\color{black}Suppose this} leader block is mined during round $R$. If an honest transaction enters a block {\color{black}and the block} is broadcast by round $R$, then every honest miner's final ledger generated by $\textbf{LedSeq}_l(r)$ will include this honest transaction.
\end{lemma}
\begin{proof}
Suppose the honest transaction $tx$ enters block $B$ which is broadcast by round $R$.  Note that  $B$ {\color{black}may} be honest or adversarial, a voter block or a leader block, and it can be on the main blockchain or an orphan block. Denote the honest leader block on level $l$ as $B_l$.

By saying block $B$ is reachable from block $A$, we mean $A$ can points to $B$ by a sequence of reference links.
According to the Prism protocol, all blocks which are reachable from an honest leader block will be included in the final ledger. By round $R$, one of the following three cases must be true:

1) $B$ is not reachable by any blocks. According to the Prism protocol, $B_l$ will reference $B$, {\color{black}so $B$ will be included in} the final ledger.

2) $B$ is reachable from an honest leader block whose level is smaller than $l$, then $B$ {\color{black}must already be} included {\color{black}in} the final ledger.

3) $B$ is {\color{black}reachable from} some block(s), but none of these block(s) is an honest leader block {\color{black}whose level is smaller than $l$}. Note that the number of proposer blocks by round $R$ is finite, and that reference links {\color{black}cannot} form a circle. Thus, among all the proposer blocks which can reach $B$, there must be {\color{black}at least} one proposer block which is not referenced by any other block by round $R$. Denote {\color{black}such a} block as $B_r$. Then according to the Prism protocol, $B_l$ will reference $B_r$. {\color{black}As a sequence, $B$ will be included in the final ledger}.

{\color{black}Once} $B$ is included in the ledger, the {\color{black}honest transaction $ tx$  will not be discarded}.
\end{proof}

%----------------------------------------------------------
\begin{theorem} \label{thm: transaction becomes permanent}
For {\color{black}every} $\epsilon >0$ and {\color{black}every integer}
\begin{align} \label{def: r permanent tx}
r \ge  \frac{25}{(1-\ep)^2\xi^2\eta} \log \frac{24m\eta^{-2}}{\epsilon},
\end{align}
every transaction that is on an honest blockchain $r$ rounds after its block is broadcast
%an honest transaction that enters into a block
is $\epsilon$-permanent. % {\color{black}$r$ rounds after the block is broadcast}.
\end{theorem}

\begin{proof}
\begin{figure}\centering
\includegraphics[clip, trim=0cm 2cm 10cm 0cm,width=0.7\columnwidth]{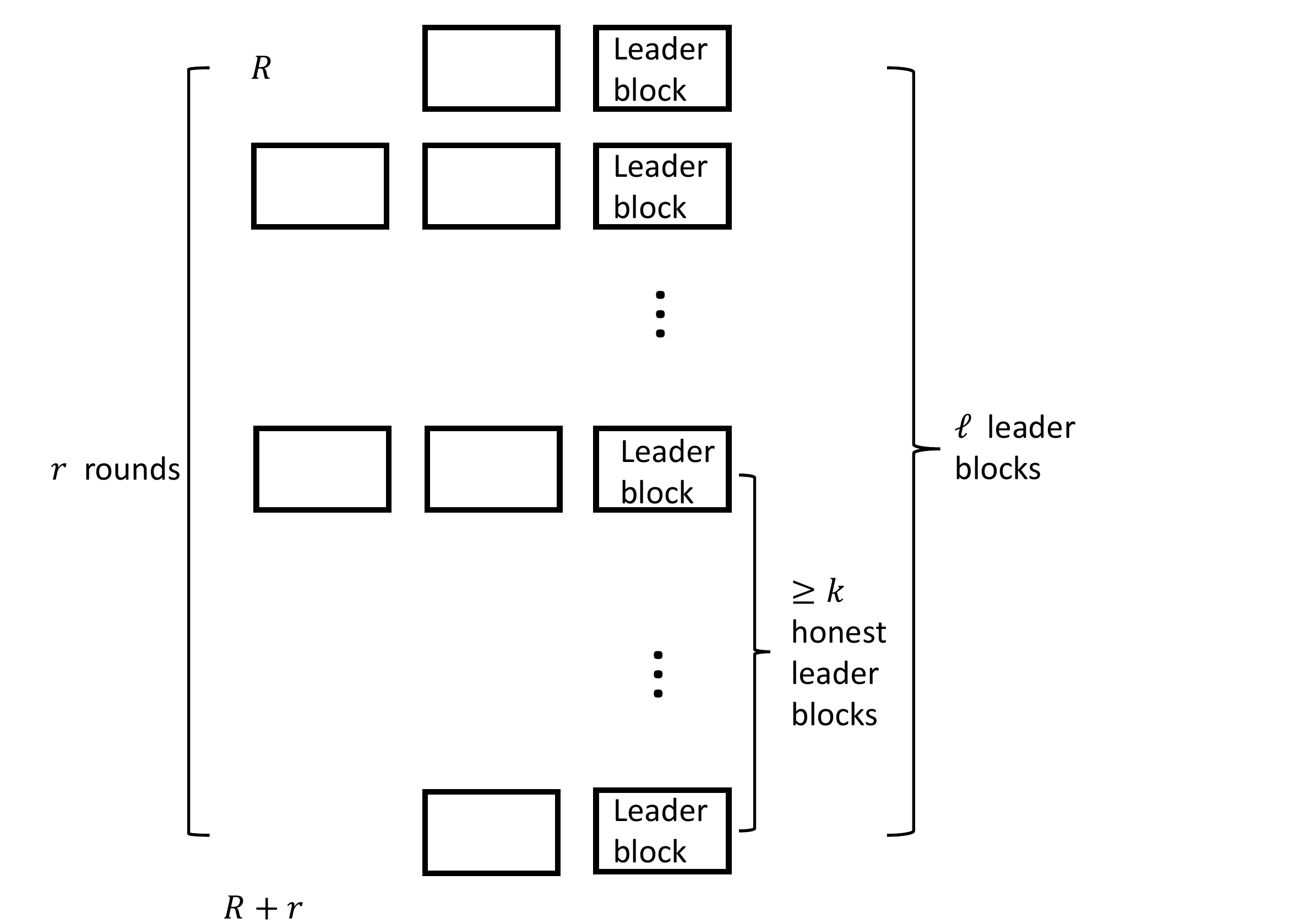}
\caption{Illustration to prove Theorem \ref{thm: transaction becomes permanent}}
\end{figure}
Let
\begin{align}
\ell & = {\color{black} \left\lceil (1-\ep)qr \right\rceil}\label{12.91} \\
k & = {\color{black}\left\lfloor \frac{\xi}{2}\ell \right\rfloor} \label{12.901}\\
w & = \left\lfloor \frac{\ell}{2q} \right\rfloor \label{12.92}\\
u & =  \left\lfloor \frac{k}{2q} \right\rfloor. \label{12.94}
\end{align}
{\color{black}
Let $R$ be the round during which the block including the honest transaction is broadcast.}
Define
\begin{align}
  G = G_0[R+r-u, R+r]  \cap G_0[R+r-w, R+r] \cap G_0[R, R+r].
\end{align}

Note that 1) According to Theorem \ref{thm: blockchain j growth} and \eqref{12.91}, under $G_0[R,R+r]$, the proposer blockchain grows by {\color{black}at least} $\ell$ leader blocks {\color{black}during rounds $\{R,\ldots, R+r\}$}. 2) According to Theorem \ref{thm: blockchain quality for proposer block}, under event $G_0[R+r-w, R+r]$, by round $R+r$ the last $\ell$ leader blocks includes at least $\epp{2}$ fraction of honest ones. Since $k \le \epp{2}\ell$, at least $k$ out of the last $\ell$ leader blocks are honest. 3) According to Lemma \ref{lemma: blockchain j growth2 at least rounds}, under event
$G_0[R+r-u, R+r]$, the deepest one of these $k$ honest leader blocks is mined at least $\frac{k}{2q}$ rounds {\color{black}before} round $R+r$. 4) We have
\begin{align}
    \frac{k}{2q} & \ge \frac{1}{2q}\left\lfloor \frac{\xi}{2}\ell \right\rfloor \label{14.001}\\
    & \ge  \frac{1}{2q} \left\lfloor \frac{\xi}{2}(1-\ep)qr \right\rfloor \label{14.002}\\
    & \ge \frac{1}{2q} \left\lfloor \frac{\xi}{2}(1-\ep)q\frac{25}{(1-\ep)^2\xi^2\eta} \log \frac{24m\eta^{-2}}{\epsilon} \right\rfloor\\
    & \ge  \frac{1}{2q} \left\lfloor \frac{25q}{2(1-\ep)\xi \eta} \log \frac{24m\eta^{-2}}{\epsilon} \right\rfloor \label{14.003}\\
    & > \frac{1}{2q} \left( \frac{25q}{2(1-\ep)\xi\eta} \log \frac{24m\eta^{-2}}{\epsilon} -1 \right) \label{14.004}\\
    & > \frac{1}{2q} \left( \frac{10q}{(1-\ep)\xi\eta} \log \frac{24m\eta^{-2}}{\epsilon} \right) \label{14.005}\\
    & = \frac{5}{(1-\ep) \xi \eta}\log \frac{{\color{black}12}m\eta^{-2}}{{\color{black}\frac{\epsilon}{2}}} \label{14.006},
\end{align}
where \eqref{14.001} is due to {\color{black}\eqref{12.901}}, \eqref{14.002} is due to {\color{black}\eqref{12.91}}, and \eqref{14.005} is obvious due to $\xi \in (0,1]$.
According to Theorem \ref{thm: permanent leader block} and \eqref{14.006},  the deepest honest leader block is $\frac{\epsilon}{2}$-permanent after round $R+r$  under event $G$. Next, we will lower bound probability of $G$.

Note that
\begin{align}
u   & \le \frac{k}{2q} \\
    & \le \frac{\xi \ell}{4q} \label{15.00}\\
    & < \frac{\ell}{2q} - 1 \label{15.01} \\
    & < w, \label{15.016}
\end{align}
where \eqref{15.00} is due to \eqref{12.901}, \eqref{15.01} is due to $q\le \ep$, and \eqref{15.016} is due to \eqref{12.92}. Also,
\begin{align}
w & \le \frac{\ell}{2q} \label{15.0171} \\
& \le  {\color{black}r}, \label{15.025}
\end{align}
where \eqref{15.0171} is due to \eqref{12.92} and
\eqref{15.025} is due to {\color{black}\eqref{12.91}}.
We have $u < w < s$. According to definition, $G_0[R+r-u, R+r] \subset  G_0[R+r-w, R+r] \subset G_0[R,R+r]$. Then,
\begin{align}
P(G) = & P \left( G_0[R+r-u, R+r]\right)  \\
> & 1- {\color{black}5\eta^{-2}}e^{- \eta u} \label{13.-1} \\
> & 1- {\color{black}5\eta^{-2}}e^{- \eta (\frac{k}{2q}-1)} \label{13.0}\\
{\color{black}\ge} & 1- {\color{black}5\eta^{-2}}e^{-\frac{5}{(1-\ep) \xi}\log \frac{24m\eta^{-2}}{\epsilon} + \eta} \label{13.1}\\
> & 1- {\color{black}5\eta^{-2}}e^{-\log \frac{10m\eta^{-2}}{\epsilon}} \label{13.3}\\
\ge &  1- {\color{black}5\eta^{-2}}e^{-\log \frac{10\eta^{-2}}{\epsilon}} \label{13.4}\\
= & 1- \frac{\epsilon}{2},
\end{align}
where \eqref{13.-1} is due to Lemma \ref{lemma: prob of typical event j},
\eqref{13.0} is due to \eqref{12.94},
\eqref{13.1} is due to {\color{black}\eqref{14.006}}, \eqref{13.3} is due to $0<\xi\le1$, and {\color{black}\eqref{13.4} is due to $m\ge 1$.}
According to the union rule, the deepest honest leader block is $\epsilon$-permanent after round $R+r$.
{\color{black}According to Lemma \ref{lemma: honest leader block include blocks}, the honest transaction will become a $\epsilon$-permanent transaction after round $R+r$.}
\end{proof}

\section{Bounded-delay model}
\subsection{The bitcoin protocol in bounded-delay model} \label{sect: bounded-delay, bitcoin}
In this section we generalize the results to the bounded-delay model, in which there is an upper bound $T$ on the delay for message delivery. That is to say, if a block  is broadcast to the network during round $r$, by round $r+T$, all other miners would have received the block. In the special case of $T=1$, this model degenerates to the synchronous model described in Section \ref{sec: model and definitions}.

Let $H[r]$, $X[r]$, $Y[r]$  and $Z[r]$ be defined in the same way as that in synchronized model. In particular, $H[r]$ stands for the number of honest blocks mined during round $r$ for $r=1,2,\ldots$.
A round is said to be a $T$-left-isolated successful round if a single honest block is mined during this round and no other honest block is mined in the previous $T-1$ rounds.  Accordingly, we define the following indicators for $r = T, T+1, T+2, \ldots$:  % let $X'[r]$ indicates whether $r$ is an isolated uniquely successful round or not:
\begin{align}\label{def: X'[r]}
    X'[r]= \begin{cases}
    1, \quad \;\; \text{if}\; H[r]=1\;\text{and}\; H[r-1] = H[r-2] = \ldots = H[r-T+1] = 0, \\
    0, \quad \;\; \text{otherwise}.
    \end{cases}
\end{align}

\begin{proposition} \label{prop: X'}
For every $r= T, T+1, \ldots$,
\begin{align}
    \mathbb{E}[X'[r]] & > q(1-q)^{T}  \label{a 0.184}\\
    &> q(1-Tq).\label{a 0.185}
\end{align}
\end{proposition}
\begin{proof}
\eqref{a 0.184} is due to \eqref{equ: E[X]=q} and Proposition \ref{prop: Y} and \eqref{a 0.185} is due to Bernoulli's inequality \eqref{equ: bernoulli}.
\end{proof}

A round is said to be a $T$-doubly-isolated successful round if a single honest block is mined during this round and no other honest block is mined within $T-1$ rounds before or after the round.  Accordingly, we define the following indicators for $r= T, T+1, T+2, \ldots$:
%, let $Y'[r]$ indicates whether $r$ is a $T$-isolated uniquely successful round or not:
\begin{align}\label{def: Y'[r]}
    Y'[r]= \begin{cases}
    1, \quad \;\; &\text{if}\;  H[r]=1, H[r-1] = H[r-2] = \ldots = H[r-T+1]=0, \\
    & \;\text{and}\; H[r+1] = H[r+2] = \ldots = H[r+T-1] = 0, \\
    0, \quad \;\; &\text{otherwise}.
    \end{cases}
\end{align}

\begin{proposition} \label{prop: Y'}
For every $r = T, T+1, \ldots$,
\begin{align}
    \mathbb{E}[Y'[r]] & > q(1-q)^{2T-1} \label{equ: Y'} \\
    & > q(1-(2T-1)q). \label{a 0.003}
\end{align}
\end{proposition}
\begin{proof}
Note that $H[r]=1$ indicates $Y[r]=1$, $H[r]=0$ indicates $X[r]=0$. Then, \eqref{equ: Y'} is due to \eqref{equ: E[X]=q} and Proposition \ref{prop: Y}, and \eqref{a 0.003} is due to Bernoulli's inequality \eqref{equ: bernoulli}.
\end{proof}

From this point onward,
it is assumed that the mining difficulty is adjusted to be sufficiently low such that the probability that one or more honest blocks are mined in a slot satisfies
\begin{align} \label{equ: honest majority asyn}
    q \le \epp{20T}.
\end{align}
\begin{proposition} \label{prop: Z < X'}
For every $r = T, T+1, \ldots$,
\begin{align}
     \mathbb{E}[Z[r]] < \mathbb{E}[X'[r]].
\end{align}
\end{proposition}
\begin{proof}
\begin{align}
    \mathbb{E}[Z[r]] & = pt \label{a 0.0}\\
    & = \frac{t}{n-t}p(n-t) \label{a 0.1} \\
    & = (1-\xi)p(n-t) \label{a 0.2}\\
    & < (1-\xi)\frac{q}{1-q} \label{a 0.3}\\
    & < q(1-Tq) \label{a 0.4}\\
    & \le \mathbb{E}[X'[r]], \label{a 0.5}
\end{align}
where \eqref{a 0.2} is due to \eqref{def: xi},
\eqref{a 0.3} is due to Proposition \ref{prop: X}, \eqref{a 0.4} is due to \eqref{equ: honest majority asyn}, and \eqref{a 0.5} is due to Proposition \ref{prop: X'}.
\end{proof}

For $T\le s < r$, we define $x(\cdot)$ on $\mathcal{R}^{r-s+T-1}$ by
\begin{align}
    x(h_{s-T+1}, \ldots, h_{r-1}) = \sum_{i=s}^{r-1}\mathds{1}\{h_i=1, h_{i-1}=\ldots = h_{i-T+1}=0\}.
\end{align}
Likewise,  we define $y(\cdot)$ on $\mathcal{R}^{r-s+2T-2}$ by
\begin{align}
    y(h_{s-T+1}, \ldots, h_{r+T-2}) = \sum_{i=s}^{r-1}\mathds{1}\{h_i=1, h_{i-T+1}=\ldots = h_{i-1} = h_{i+1} = \ldots = h_{i+T-1}=0\}.
\end{align}
Although $h_i$ is allowed to take arbitrary real values, the indicator function yields binary value.

For all integers $s$ and $r$ satisfying $T\le s < r$,
% similarly to $X[s,r]$ and $Y[s,r]$,
we define
\begin{align}
    X'[s,r] & = x(H[s-T+1], \ldots, H[r-1]) \\
             & = \sum_{i=s}^{r-1}X'[i] \label{a 0.400}
\end{align}
and
\begin{align}
    Y'[s,r] &= y(H[s-T+1], \ldots, H[r+T-2]) \\
            &= \sum_{i=s}^{r-1}Y'[i] . \label{a 0.401}
\end{align}

\begin{comment}
Note that \begin{align}
\mu' & = 1-\frac{1+\epp{40}}{(1-\epp{20})(1-q)^T} (1-\xi) - \frac{\epp{20}}{1-\epp{20}}(1+\epp{40}) \\
& > 1-\frac{1+\epp{40}}{(1-\epp{20})(1-qT)} (1-\xi) - \frac{\epp{20}}{1-\epp{20}}(1+\epp{40}) \label{a 0.51}\\
& \ge  1-\frac{1+\epp{40}}{(1-\epp{20})^2} (1-\xi) - \frac{\epp{20}}{1-\epp{20}}(1+\epp{40}) \label{a 0.52}\\
& > 0, \label{a 0.53}
\end{align}
\begin{align}
\eta'  & = (1-\epp{10})(1-q)^{T-1} \label{def: eta'} \\
    \mu' & = 1-\frac{1+\epp{40}}{(1-\epp{20})(1-q)^T} (1-\xi) - \frac{\epp{20}}{1-\epp{20}}(1+\epp{40}) \label{def: mu'}
\end{align}
where \eqref{a 0.51} is due to Bernoulli's inequality, \eqref{a 0.52} is due to \eqref{equ: honest majority asyn} $q\le \epp{20T}$, \eqref{a 0.53} is due to $\xi \in [0,1)$.
\end{comment}

%-------------------------------------------------------------------------

\begin{definition} \label{def: F bounded-delay}
For all integers $T\le s < r$, define event
\begin{align}
    F[s,r] := F_1[s,r] \cap F_2[s,r] \cap F_3[s,r] \cap F_4[s,r]
\end{align}
where
\begin{align}
    F_1[s,r] & := \left\{(1-\epp{20})\mathbb{E}[X'[s,r]] < X'[s,r]\right\},\label{F1}\\
    F_2[s,r] & :=  \left\{ X[s,r] < (1+\epp{20})\mathbb{E}[X[s,r]] \right\}, \label{F2} \\
    F_3[s,r] & := \left\{ (1-\epp{20})\mathbb{E}[Y'[s,r]] < Y'[s,r]\right\}, \label{F3}\\
    F_4[s,r] & :=  \left\{Z[s,r] < \mathbb{E}[Z[s,r]] + \epp{20}\mathbb{E}[X'[s,r]]\right\}.\label{F4}
\end{align}
\end{definition}

%-----------------------------------------------------------------------------
\begin{definition}
 Let $f$ be a function on $\mathcal{R}^n$. Let $x, x'\in \mathcal{R}^n$.  A function $f(x_1, x_2, \ldots, x_n)$ is $k$-Lipschitz if $|f(x)-f(x')| \le k$ whenever $x$ and $x'$ differ in at most one coordinate.
\end{definition}

\begin{theorem}\label{thm: lipschitz}
(McDiarmid’s inequality, \cite[page 40]{vershynin2018high}) If $f$ on $\mathcal{R}^n$ is $k$-Lipschitz and $X_1,\ldots,X_n$ are independent random variables, then for every $t>0$,
\begin{align}
    P(f(X_1, X_2, \ldots, X_n)>\mathbb{E}(f(X_1, X_2, \ldots, X_n)) + t) & \le e^{-\frac{2t^2}{nk^2}}, \\
    P(f(X_1, X_2, \ldots, X_n)<\mathbb{E}(f(X_1, X_2, \ldots, X_n)) - t) & \le e^{-\frac{2t^2}{nk^2}}.
\end{align}
\end{theorem}

\begin{lemma}\label{lemma: X' lip}
For $T\le s < r$, $x(\cdot)$ is $1$-Lipschitz and $y(\cdot)$ is $2$-Lipschitz.
\end{lemma}
\begin{proof}
 Define $x_i = \mathds{1}\{h_i=1, h_{i-1}=\ldots = h_{i-T+1}=0\}$, then $x(h_{s-T+1}, \ldots, h_{r-1}) = \sum_{i=s}^{r-1} x_i$. Suppose  $h_{s-T+1}, \ldots, h_k, \ldots, h_{r-1}$ changes to $h_{s-T+1}, \ldots, h'_k, \ldots, h_{r-1}$.
 Let $\ell$ be equal to the smaller one of $k+T-1$ and $r-1$.
Only $x_k, x_{k+1}, \ldots, x_{\ell}$  may be affected by this change. By definition, at most one of $x_{k+1}, \ldots, x_{\ell}$ can be non-zero. Then there are two cases before the change: 1) All of $x_{k+1}, \ldots, x_{\ell}$ are equal to $0$. In this case, the change of $h_k$ can change (increase or decrease) the value of $x_{k}$ by at most $1$, but has no impact on $x_{k+1}, \ldots, x_{\ell}$. 2) There exists a $j$ between $k+1$ and $\ell$. In this case, $h_k$ must be zero according to the definition of $x_j$. Thus, $h'_k\ne 0$, $x_j$ changes from $1$ to $0$. Meanwhile, $x_k$ may change from $0$ to $1$ or remain zero, so $x_k + x_j$ is not going to differ by more than $1$ from of its original value. In either case, $x(h_{s-T+1}, \ldots, h_{r-1})$ can change by no more than $1$, so $x(\cdot)$ is $1$-Lipschitz.

Define $y_i = \mathds{1}\{h_i=1, h_{i-T+1}=\ldots = h_{i-1} = h_{i+1} = \ldots = h_{i+T-1}=0\}$, then $y(h_{s-T+1}, \ldots, h_{r+T-2})  = \sum_{i=s}^{r-1} y_i$. Suppose  $h_{s-T+1}, \ldots, h_k, \ldots, h_{r+T-2}$ changes to $h_{s-T+1}, \ldots, h'_k, \ldots, h_{r+T-2}$.
Let $m$ be the larger one of $k-T+1$ and $s-T+1$, let  $n$ be the smaller one of $k+T-1$ and  $r+T-2$. Only $y_{m}, \ldots, y_{k-1}, y_{k}, y_{k+1}, \ldots, y_{n}$  are possibly affected by this change.
By definition, at most two elements of $y_{m}, \ldots, y_{n}$ can be non-zero, and they must be on different sides of $y_k$. Then there are two cases: 1) There are no more than one none-zero elements in $y_m,\ldots, y_n$, in this case changing  $h_k$ can not change the value of $g$ by more than $2$. 2) There exists an $p$ and $q$ satisfying $m\le p \le k-1$, $k+1\le q \le n$ such that $y_p = 1$ and $y_q =1$. In this case, we must have $h_k=0$ and $h'_k\ne 0$ according to the definition of $y_p, y_q$. Thus $y_p$ and $y_q$ change from $1$ to $0$. Meanwhile,  $y_k$ may change from $0$ to $1$ or remain unchanged, and $y_p + y_k + y_q$ can change by $1$ or $2$, but not more than $2$ from of its original value. So $y(\cdot)$ is $2$-Lipschitz.
\end{proof}

We define
\begin{align}
    \eta' & = \eps{4000T^2}q^2(1-q)^{4T-2}. \label{def: gamma'}
\end{align}

\begin{lemma}
For all integers $T\le s < r$,
\begin{align}
    P(F[s,r]) > 1-4e^{-\eta'(r-s)},
\end{align}
where $\eta'$ is given in \eqref{def: gamma'}.
\end{lemma}

\begin{proof}
Because $X'[r]$ and $X'[s]$ are dependent, standard Chernoff bound does not apply. Similarly for $Y'[r]$ and $Y'[s]$. However, due to Lemma \ref{lemma: X' lip}, we have
\begin{align}
    P(F_1^c[s,r]) & =  P \left( X'[s,r] \le \mathbb{E}[X'[s,r]] - \epp{20}\mathbb{E}[X'[s,r]] \right)
    \\
    & \le e^{-\eps{200(r-s+T-1)}\mathbb{E}[X'[s,r]]^2} \label{a 0.9} \\
    & \le e^{-\eps{200}q^2(1-q)^{2T}(r-s+T-1)} \label{a 1.0}\\
    & \le e^{-\eps{200}q^2(1-q)^{2T}(r-s)}  \label{a 1.001}
\end{align}
where \eqref{a 0.9} is due to Theorem \ref{thm: lipschitz}, \eqref{a 1.0} is due to Proposition \ref{prop: X'}, and \eqref{a 1.001} is due to $T\ge 1$.

Similarly,
\begin{align}
    P(F_3^c[s,r]) & = P \left( Y'[s,r] \le \mathbb{E}[Y'[s,r] -\epp{20}\mathbb{E}[Y'[s,r]] \right) \\
    & \le e^{-\eps{200(r-s+2T-2)}\mathbb{E}[Y'[s,r]]^2}\label{a 2.0}\\
    & \le e^{-\eps{200}q^2(1-q)^{4T-2}(r-s+2T-2)} \label{a 2.1}\\
     & \le e^{-\eps{200}q^2(1-q)^{4T-2}(r-s)} \label{a 2.111}
\end{align}
where \eqref{a 2.0} is due to Theorem \eqref{thm: lipschitz}, \eqref{a 2.1} is due to \eqref{equ: Y'}, and \eqref{a 2.111} is due to $T\ge 1$.

By Proposition \ref{prop: Chernoff bound},
\begin{align}
    P(F_2^c[s,r]) & = P\left(X[s,r] \ge  (1+\epp{20})\mathbb{E}[X[s,r]]\right)\\
    &\le e^{-\eps{1200}q(r-s)}. \label{a 2.112}
\end{align}

According to \eqref{a 0.5}, $E[Z[s,r]] < E[X'[s,r]]$. Note that the moment generating function for binomial random variable $Z[r]\sim Binomial(t,p)$ is $(1-p+pe^{u})^{t}$ \cite{das1989statistical}. Pick arbitrary $u>0$. We have
\begin{align}
P(F_4^c[s,r])& = P\left(Z[s,r] \ge \mathbb{E}[Z[s,r]] + \epp{20T} \mathbb{E}[X'[s,r]]\right)  \\
& \le P\left(Z[s,r] \ge \mathbb{E}[Z[s,r]] + \epp{40T} \mathbb{E}[Z[s,r]] + \epp{40T} \mathbb{E}[X'[s,r]]\right) \\
& <  \frac{\mathbb{E}[e^{Z[s,r]u}]}{e^{(1+\epp{40T})\mathbb{E}[Z[s,r]]u + \epp{40T}\mathbb{E}[X'[s,r]]u}} \label{ahe: 1.9}\\
& = \frac{(1-p+pe^u)^{t(r-s)}}{e^{(1+\epp{40T})(r-s)tpu + \epp{40T}q(1-q)^{T}u(r-s)}} \label{ahe: 2.0}\\
& \le e^{(e^u-1-u(1+\epp{40T}))tp(r-s)-\epp{40T}uq(1-q)^{T}(r-s)}, \label{ahe: 2.1}
\end{align}
where \eqref{ahe: 1.9} is by the Chernoff inequality and \eqref{ahe: 2.1} is due to $1+x \le e^x$ for every $x\ge 0$ (here $x = p(e^u-1)$). Let $u=\log(1+\epp{40T})$. Then
\begin{align}
P(F_4^c[s,r])& \le e^{(\epp{40T} - (1+\epp{40T})\log(1+\epp{40T}))tp(r-s) - \epp{40T}\log(1+\epp{40T})q(1-q)^{T}(r-s)} \\
& < e^{ - \epp{40T}\log(1+\epp{40T})q(1-q)^{T}(r-s)}\label{a 2.10} \\
& < e^{-\eps{4000T^2}q(1-q)^{T}(r-s)} \label{a 2.11},
\end{align}
where \eqref{a 2.10} is due to $(1+x)\log(1+x)>x$ for all $x>0$ and \eqref{a 2.11} is due to $\log(1+\epp{40T})>\epp{100T}$ for all $\xi \in (0,1]$ and $T\ge 1$.

Since $\eta'$ defined in \eqref{def: gamma'} dominates the corresponding exponential coefficients in \eqref{a 1.001}, \eqref{a 2.111}, \eqref{a 2.112}, and \eqref{a 2.11}, we have
\begin{align}
P(F[s,r]) & = 1-P(F^c[s,r]) \\
& \ge 1- P(F_1^c[s,r]) - P(F_2^c[s,r]) - P(F_3^c[s,r]) - P(F_4^c[s,r]) \\\
& > 1-4e^{-\eta' (r-s)}.
\end{align}
\end{proof}

%---------------
%-------------------------------------------------

\begin{lemma} \label{lemma: asyn properties}
For all integers $T\le s < r-\frac{2}{q}$ , the following inequalities hold under event $F[s,r]$:
\begin{align} \label{equ: X'>(1-delta/20)q(1-q)^T-1s}
    (1-\epp{20})q(1-q)^{T}(r-s) < X'[s,r]
\end{align}
\begin{align} \label{equ: X <(1+epp20)q(r-s)}
    X[s,r] < (1+\epp{20})q(r-s)
\end{align}
\begin{align}
    (1-\epp{3})q(r-s) <Y'[s,r] \label{a 3.0}
\end{align}
\begin{align}
    Z[s,r]< (1-\epq{2}{3})q(r-s) \label{equ: asyn Z<(1-2/3beta)qs}
\end{align}
\begin{align} \label{equ: asyn Z<(1-delta/2)X'}
    Z[s,r+T] < (1-\epp{2})X'[s,r]
\end{align}
\begin{align} \label{equ: Z<Y'}
    Z[s-T,r+T] < Y'[s,r] .
\end{align}
\end{lemma}
\begin{proof}
Under $F[s,r]$, \eqref{equ: X'>(1-delta/20)q(1-q)^T-1s} follows directly from \eqref{F1}. \eqref{equ: X <(1+epp20)q(r-s)} follows directly from \eqref{F2}.

To prove \eqref{a 3.0}, we write
\begin{align}
     Y'[s,r] & >  (1-\epp{20})E[Y'[s,r]] \label{a 4.-3} \\
     & > (1-\epp{20})(1-(2T-1)q)q(r-s) \label{a 4.-2}\\
     & > (1-\epp{20})(1-\epp{10})q(r-s) \label{a 4.-1}\\
     & > (1-\epp{3})q(r-s), \label{a 4.0}
\end{align}
where \eqref{a 4.-3} is due to \eqref{F3}, \eqref{a 4.-2} is due to Proposition \ref{prop: Y'}, \eqref{a 4.-1} is due to \eqref{equ: honest majority asyn}, and \eqref{a 4.0} is due to $\xi \in [0,1)$ and $T\ge 1$.

To prove \eqref{equ: asyn Z<(1-2/3beta)qs},
\begin{align}
 Z[s,r] & < \mathbb{E}[Z[s,r]] + \epp{20}\mathbb{E}[X'[s,r]] \label{a 5.-1}\\
 & = pt(r-s) + \epp{20}q(1-q)^{T}(r-s) \label{a 5.0} \\
 & = \frac{t}{n-t}(n-t)p(r-s) + \epp{20}q(1-q)^{T}(r-s) \label{a 5.1}\\
 & \le  (1-\xi)\frac{q}{1-q}(r-s) + \epp{20}q(1-q)^{T}(r-s) \label{a 5.2}\\
 & \le (1-\xi)\frac{q}{1-\epp{20T}}(r-s) + \epp{20}q(r-s) \label{a 5.3}\\
 & < \left(1-\epq{2}{3}\right)q(r-s), \label{a 5.4}
\end{align}
where \eqref{a 5.-1} is due to  \eqref{F4}, \eqref{a 5.2} is due to Proposition \ref{prop: X}, \eqref{a 5.3} is due to \eqref{equ: honest majority asyn}, and \eqref{a 5.4} is due to $\xi \in (0,1]$.

By assumption \eqref{equ: honest majority asyn},
\begin{align}
    T & \le \frac{\xi}{20q} \\
    & < \epp{40}(r-s) \label{a 6.0}
\end{align}
where \eqref{a 6.0} is by the assumption of this lemma. Thus,
\begin{align}
    r-s+T < (1+\epp{20})(r-s). \label{a 6.09}
\end{align}
To prove \eqref{equ: asyn Z<(1-delta/2)X'}, we begin with \eqref{a 5.2} with $r$ replaced by $r+T$:
\begin{align}
    Z[s, r+T] & < (1-\xi)\frac{q}{1-q}(r-s+T) + \frac{\xi}{20}q(1-q)^{T}(r-s+T) \label{a 6.1} \\
    & < \left( \frac{1-\xi}{(1-q)^{T+1}} + \epp{20}\right)(1+\epp{40}) q(1-q)^T(r-s) \label{a 6.11}\\
      & < \left( \frac{1-\xi}{1-(T+1)q} + \epp{20}\right)(1+\epp{40}) q(1-q)^T(r-s) \label{a 6.13}\\
    & < \left(\frac{1-\xi}{1-(T+1)q} + \epp{20}\right)(1 + \epp{40})\frac{X'[s,r]}{1-\epp{20}} \label{a 6.4} \\
    & < \left(\frac{1-\xi}{1-\epp{10}} + \epp{20}\right)(1 + \epp{40})\frac{1}{1-\epp{20}}X'[s,r] \label{a 6.5}\\
    & < (1-\epp{2})X'[s,r] \label{a 6.6}
\end{align}
where \eqref{a 6.11} is due to \eqref{a 6.09},  \eqref{a 6.13} is due to \eqref{equ: bernoulli}, \eqref{a 6.4} is due to \eqref{equ: X'>(1-delta/20)q(1-q)^T-1s}, \eqref{a 6.5} is due to $q < \epp{10(T+1)}$, \eqref{a 6.6} is due to $\xi \in [0,1)$.

To prove \eqref{equ: Z<Y'},
\begin{align}
    Z[s-T,r+T] & < (1-\epq{2}{3})q(r-s+2T) \label{a 7.1} \\
    & < (1-\epq{2}{3})(1+\epp{20})q(r-s) \label{a 7.3}\\
    & < (1-\epp{3})q(r-s) \label{a 7.4}\\
    & < Y'[s,r], \label{a 7.5}
\end{align}
where \eqref{a 7.1} is due to \eqref{equ: asyn Z<(1-2/3beta)qs}, \eqref{a 7.3} is due to \eqref{a 6.09}, \eqref{a 7.4} is due to $\xi \in [0,1)$, and \eqref{a 7.5} is due to \eqref{a 3.0}.
\end{proof}

%------------------------------------------------------------------------------------------

\begin{definition}
For all integers $T\le s<r-\frac{2}{q}$, define typical event
\begin{align}
    J[s,r] \coloneqq \cap_{0\le a \le s-T, b\ge 0}F[s-a, r+b].
\end{align}
\end{definition}
$J[s,r]$ occurs when events $F[s-a, r+b]$ simultaneously occurs for all $a, b$, i.e., the  ``$F$'' event occurs over all intervals containing $[s,r]$. Like event $G$, the event $F$ represents a collection of outcomes that constrain the number of blocks mined in all intervals that contain $[s,r]$, including arbitrarily large intervals that end in the arbitrarily far future. Intuitively,
 we define $J[s,r]$ to allow the ``good'' properties in Lemma \ref{lemma: asyn properties} to extend to all intervals containing $[s,r]$ under the event.

%-------------------------------------------------------------------

\begin{lemma} \label{lemma: asyn prob. typical event}
For all integers $T\le s < r -\frac{2}{q}$,
\begin{align}
    P(J[s,r]) > 1 - 5\eta'^{-2}e^{-\eta'(r-s)}. \label{a 7.55}
\end{align}
\end{lemma}
\begin{proof}
Due to the stationarity of processes $X, Y, Z, X'$, and $Y'$, $P(F[s,r]) = P(F[T,T+r-s])$ for all $s, r$. Evidently the probability  depends on $r$ and $s$ only through the interval length $r-s$:
\begin{align}
    P(J^c[s, r]) & = P(\cup_{0\le a \le s-T, b\ge 0}F^c[s-a,r+b]) \\
    & = P(\cup_{0\le a \le s-T, b\ge 0}F^c[T,r-s+a+b+T])\\
    & \le \sum_{0\le a \le s-T, b\ge 0}P(F^c[T,r-s+a+b+T]) \\
    & = \sum_{k=0}^{\infty} \sum_{0\le a \le s-T, b\ge 0,a+b=k}P(F^c[T,r-s+k+T])\\
    & < \sum_{k=0}^{\infty} (k+1)P(F^c[T,r-s+k+T])\\
    & < \sum_{k=0}^{\infty}(k+1)4e^{-\eta'  (r-s+k)}\\
    & = 4e^{-\eta' (r-s)}\sum_{k=0}^{\infty} (k+1) e^{-\eta' k}\\
    & = \frac{4}{(1-e^{-\eta'})^2}e^{-\eta'(r-s)}. \\
\end{align}
According to \eqref{equ: honest majority asyn} and \eqref{def: gamma'}, $\eta' < \frac{1}{4000}$. Thus, \eqref{a 7.55} is established using the fact that $1-e^{-x} > \frac4{\sqrt{5}}x$ for all $x\in [0, \frac{1}{4000}]$.
\end{proof}
% ( 1-exp(-1/4000) ) * 4000 = 0.9998750104158738

%--------------------------------------------------------------

\begin{lemma} \label{lemma: asyn pre common prefix}
Suppose some blockchain's $k$th block $B$ is mined by an honest miner during round $r>T$ in a $T$-doubly-isolated successful round. Then, after round $r+T$, the $k$th block of every blockchain is either $B$ or an adversarial block.
\end{lemma}
\begin{proof}
Suppose, contrary to the claim, the $k$th block of another blockchain is an honest block $B'\ne B$. Let $\underline{r}$ (respectively, $\overline{r}$) denote the round number that earlier (respectively, later) of $B$ and $B'$ is mined. Then we have $\overline{r} > \underline{r}+T$ by assumption that $B$ is mined in a uniquely successful round.
Because the propagation delay is bounded by $T$, then by round $\overline{r}$ all miners have seen a block of height $k$, so no other honest block of height $k$ will be mined. This contradicts the assumption that $B$ and $B'$ are both at position $k$.
Hence the proof of Lemma \ref{lemma: asyn pre common prefix}.
\end{proof}

%--------------------------------------------------------------------

\begin{lemma} \label{lemma: asyn pre blockchain growth}
(Lemma 29 in \cite{garay2015bitcoin}) Let $ T \le s \le r-T$ be integers. Suppose an honest blockchain is of length $l$ by round $s$. Then by round $r$, the length of every honest blockchain is at least $l + X'[s,r-T+1]$.
\end{lemma}
\begin{proof}
By induction on $r$: Consider $r = s+T$. If $X'[s]=0$, then $X'[s,s+1]=0$.
Since propagation delays are upper bounded by $T$ rounds, all miners have seen a block of height $l$ by round $s+T$, so all honest blockchains' height is at least $l$.
If $X'[s]=1$, $X'[s,s+1]=1$. By round $s+1$, at least one honest miner would have broadcast a block of height $l+1$. Then by round $r+T$, each honest miner will adopt a blockchain of at least $l+1 = l+X[s,r-T+1]$ blocks.

We next assume the claim holds for $r= s+T, s+T+1, \ldots, s+T+u$ and show that it also holds for $s+T+u+1$. There are two cases. 1)
$X'[s+u+1] = 0$, the claim holds trivially.
2) $X'[s+u+1]=1$, then by definition of $X'$, $X'[s+u] = \ldots = X'[s+u-T+2] = 0$.
By induction, by round $s+u+1$, any honest miner's blockchain length is at least $l' = l+X'[s, s+u-T+2] = l+X'[s,s+u+1]$. Since $X'[s+u+1]=1$, by round $s+u+2$ at least one honest miner would have broadcast a blockchain of length $l'+1 = X'[s, s+u+2] + 1$. Then each honest miner will adopt a blockchain of length at least $l+X[s+u+2]$ by round $s+u+T+1$.

By induction on $r$, Lemma \ref{lemma: asyn pre blockchain growth} holds.
\end{proof}

%----------------------------------------------------------------------
\begin{lemma} \label{lemma: asyn blockchain growth}
For all integers $T \le s < r -\frac{2}{q}$ and $k\ge 2q(r-s)$, under typical event $J[s,r]$, every honest miner's $k$-deep block by round $r$ must be mined before round $s$.
\end{lemma}
\begin{proof}
The blockchain growth of an honest miner during rounds $\{s, \ldots, r-1 \}$ is upper bounded by $X[s-T, r] + Z[s-T, r]$. Under $J[s,r]$, $F[s-T,r]$ also occurs. Note that
\begin{align}
    X[s-T,r]+Z[s-T,r] & < (1+\epp{20})q(r-s+T) + (1-\epq{2}{3})q(r-s+T) \label{a 6.-1}\\
     &< (2 - \epp{2})q(r-s+T) \\
    & < 2q(r-s) \label{a 6.-05}\\
    & \le k,
\end{align}
where \eqref{a 6.-1} is due to \eqref{equ: X <(1+epp20)q(r-s)} and \eqref{equ: asyn Z<(1-2/3beta)qs}, \eqref{a 6.-05} is due to \eqref{equ: honest majority asyn} and $r-s > \frac{2}{q}$. Then, the $k$-deep block must be adopted before round $s$, thus mined before round $s$.
\end{proof}

%-----------------------------------------------------------------------

\begin{theorem} \label{thm: asyn blockchain growth}
(Blockchain growth property for bounded-delay model) Let $r, s, s_1$ be integers satisfying $T \le s_1 \le s < r-\frac{2}{q}$.  Then under typical event $J[s, r-T]$, the length of every honest blockchain must increase by at least $(1-\epp{10})(1-q)^{T}q(r-s_1)$ during rounds $\{s_1, \ldots, r\}$.
\end{theorem}
\begin{proof}
Note that $\frac{2}{q} \ge \frac{40T}{\xi} > T$ by \eqref{equ: honest majority asyn}.
Since $[s,r-T]\subset [s_1,r-T+1]$, we have
\begin{align}
    X'[s_1,r-T+1] & > (1-\epp{20})q(1-q)^{T}(r-s_1-T+1) \label{a 7.01} \\
    & = (1-\epp{20})q(1-q)^{T}(1-\frac{T-1}{r-s_1})(r-s_1) \label{a 7.11}\\
    & > (1-\epp{20})(1-\frac{T}{r-s_1})q(1-q)^{T}(r-s_1) \label{a 7.15} \\
    & \ge (1-\epp{20})(1-\epp{40})q(1-q)^{T}(r-s_1) \label{a 7.16}\\
    & > (1-\epp{10})q(1-q)^{T}(r-s_1), \label{a 7.21}
\end{align}
where \eqref{a 7.01} is due to \eqref{equ: X'>(1-delta/20)q(1-q)^T-1s}, \eqref{a 7.16} is due to $r-s_1\ge \frac{2}{q}$ and \eqref{equ: honest majority asyn} $q\le \epp{20T}$, \eqref{a 7.21} is due to $\xi \in [0,1)$.
By Lemma \ref{lemma: asyn pre blockchain growth}, the length of every honest blockchain must increase by at least $(1-\epp{10})(1-q)^{T}q(r-s_1)$ during rounds $\{s_1, \ldots, r\}$.
\end{proof}

%-------------------------------------------------------------------

\begin{theorem} \label{thm: asyn blockchain quality}
(Blockchain quality theorem for bounded-delay model) Let $r,s,k$ be integers satisfying $T \le s < r-\frac{2}{q}$ and $k\ge 2q(r-s)$. Suppose an honest miner's blockchain $j$ has more than $k$ blocks by round $r$. By round $r$, at least $\frac{\xi}{2}$ fraction of the last $k$ blocks of this miner's blockchain are honest under event $J[s, r-T]$.
\end{theorem}

\begin{proof}
Assume an honest miner adopts blockchain $C = B_0B_1\ldots B_{len(C)-1}$ by round $r$, where  $B_0$ is the genesis block and $len(C)>k$.
Let $u=len(C)-k$. Then the last $k$ blocks of $C$ are $B_u\ldots B_{len(C)-1}$. Let $B_{u'}$ be the last honest block before $B_u$. That is to say, $u' = \max \{u'|u'\le u-1, B_{u'} \text{ is honest} \}$ ($u'$ is always well defined as $B_0$ is regarded as honest).
Let $r^*$ be the round when $B_{u'}$ is mined.
By Lemma \ref{lemma: asyn blockchain growth} and the fact that $J[s, r-T] \subset J[s,r]$,  we have $r^* < s$. Let $L=len(C)-u'-1$. Note that $L\ge len(C) - u \ge k$.

Let $x$ be the number of honest blocks in $B_u\ldots B_{len(C)-1}$. To prove the theorem, it suffices to show $x > \frac{\xi}{2} k$. Since all blocks in $B_{u'+1}\ldots B_{u-1}$ are adversarial, the number of honest blocks in  $B_{u'+1}\ldots B_{len(C)}$ is also $x$. Thus, the number of adversarial blocks in $B_{u'+1}\ldots B_{len(C)-1}$ is $L - x$. Under $J[s, r-T]$, which implies $F[r^*+1, r-T]$ also occurs,
we have
\begin{align}
L-x & \le Z[r^*+1,r] \\
& < (1-\frac{\xi}{2})X'[r^*+1,r-T]  \label{a 8.00}\\
& \le (1-\frac{\xi}{2})L \label{a 8.3}\\
& \le L - \frac{\xi}{2} k, \label{a 8.4}
\end{align}
where \eqref{a 8.00} is due to \eqref{equ: asyn Z<(1-delta/2)X'}, \eqref{a 8.3} is due to Lemma \ref{lemma: asyn pre blockchain growth}, \eqref{a 8.4} is due to $L\ge k$. From  \eqref{a 8.4}, $x>\frac{\xi}{2} k$ is derived.
\end{proof}

%---------------------------------------------------------------

\begin{theorem} \label{thm: asyn common prefix}
(Common prefix property for bounded-delay model) Let $r, s ,k$ be integers satisfying $T \le s < r-\frac{2}{q}$ and $k> 2q(r-s)$. If by round $r$ an honest blockchain has a $k$-deep prefix, then the prefix is permanent after round $r$ under
$J[s+T, r-T]$.
\end{theorem}

\begin{proof}
By assumption,
\begin{align}
r-s & > \frac{2}{q} \\
& > \frac{40T}{\xi} \\
& > 2T+1.
\end{align}
Hence $J[s+T, r-T]$ is well defined.

We prove the desired result by contradiction. Suppose blockchain $C_1$ is adopted by an honest miner $P_1$ by round $r$.
Suppose, contrary to claimed, an honest miner $P_2$ first deviates from the prefix
$C_1^{\left\lceil k \right.}$ from some round number $r_2\ge r$. Specifically, $P_2$ adopts $C'_2$ by round $r_2-1$ which satisfies $C_1^{\left\lceil k \right.} \preceq C'_2$ and then adopts $C_2$ by round $r_2$ which satisfies $C_1^{\left\lceil k \right.} \npreceq C_2$.

Assume the last honest block on the common prefix of $C_2$ and $C'_2$ is mined during round $r^*$. If $r^*>0$, this common block of $C_2$ and $C'_2$ must be more than $k$ deep in $C_1$ by round $r$. According to Lemma \ref{lemma: asyn blockchain growth}, we have $r^* < s$, so that
\begin{align}
    [s+T, r-T] \subset [r^*+T+1, r_2-T]. \label{a 8.5}
\end{align}
On the other hand, if $r^*=0$, the last common block is the genesis block. Since $s\ge T$, \eqref{a 8.5} also holds.

% Note that $(r_2-T)-(r^*+T+1) \ge r-T-(s+T+1) = r-s-2T+1 > 0$.
%By assumption
Under $J[s+T, r-T]$, $F[r^*+T+1, r_2-T]$ also occurs, so that $Y'[r^*+T+1, r_2-T]>0$ according to \eqref{a 3.0}. Hence, there must be at least one $T$-doubly-isolated successful round $u\in \{r^* + T +1, \ldots, r_2-T-1\}$.
%Since $u$ is a $T$-isolated uniquely successful round, according to Lemma \ref{lemma: same length after T-1},  all honest blockchains have identical length by round $u+T$. Let $l_u$ denote one plus the length of honest miners' blockchains by round $u$.
Suppose honest miner $P$ mines $B_u$ during round $u$ with height $l_u$. According to Lemma \ref{lemma: asyn pre common prefix}, the $l_u$th block of every blockchain by round $r_2$ is either $B_u$ or adversarial block.
% Because $C_2$ and $C'_2$ are adopted by an honest miner after round $r_2-2$, they must be no shorter than $\max\{l_u: u\; \text{is a $T$-doubly-isolated successful round in}\; \{r^*+T+1, \ldots, r_2-T-1\}\}.$

%For every uniquely successful round $u$ in $\{ r^*+T+1, \ldots, r_2-T-1\}$,
If the $l_u$th block of $C_2$ and $C'_2$ are different, then at least one of them must be adversarial according to Lemma \ref{lemma: asyn pre common prefix}. On the other hand, if the $l_u$th block of $C_2$ and $C'_2$ are identical, the block must be in their common prefix, which must be adversarial by the definition of $r^*$. Thus, for each $T$-doubly-isolated successful round $u$, there is at least one adversarial block at height $l_u$. Since these adversarial blocks are mined after the last common honest block of $C'_2$ and $C_2$, they are mined after round $r^*$. Since $C'_2$ and $C_2$ are adopted by round  $r_2$, their blocks must be mined before round $r_2$. Thus, the adversarial blocks that match the $T$-doubly-isolated successful rounds are mined within $[r^*+1, r_2]$. Thus, $Z[r^*+1, r_2]\ge Y'[r^*+T+1, r_2-T]$. However, since $[s+T, r-T]\in [r^*+T+1, r_2-T]$, $F[r^*+T+1, r_2-T]$ occurs under $J[s+T, r-T]$, so that $Z[r^*+1, r_2] <  Y'[r^*+T+1, r_2-T]$ according to \eqref{equ: Z<Y'}. Contradiction arises, hence the proof of the theorem.
\end{proof}

%-----------------------------------------------------------------------

\subsection{Prism under bounded-delay model}

Recall $H_j[r]$ denotes the total number of honest blocks mined during round $r$ for blockchain $j$.  Following the definitions in Section \ref{sect: bounded-delay, bitcoin}, for $j = 0,1,\ldots, m$ and $r= T, T+1, \ldots$, we define
\begin{align}\label{def: Y'_j[r]}
    Y'_j[r]= \begin{cases}
    1, \quad \;\; \text{if}\; H_j[r]=1\;\text{and}\; H_j[s]=0 \; \text{for} \; s=r-T+1,\ldots, r-1, r+1, \ldots, r+T-1 \\
    0, \quad \;\; \text{otherwise}
    \end{cases}
\end{align}
and
\begin{align}\label{def: X'_j[r]}
    X'_j[r]= \begin{cases}
    1, \quad \;\; \text{if}\; H_j[r]=1\;\text{and}\; H_j[s]=0 \; \text{for} \; s=r-T+1, \ldots, r-1 \\
    0, \quad \;\; \text{otherwise}.
    \end{cases}
\end{align}
Basically $Y'_j[r]$ indicates whether $r$ is a $T$-doubly-isolated successful round for blockchain $j$, whereas $X'_j[r]$ indicates whether $r$ is a $T$-left-isolated uniquely successful round for blockchain $j$.

\begin{definition} \label{def: Fj bounded-delay}
For all integers $T\le s < r-\frac{2}{q}$ and $0\le j \le m$, define event
\begin{align}
    F_j[s,r] := F_{1,j}[s,r] \cap F_{2,j}[s,r] \cap F_{3,j}[s,r] \cap F_{4,j}[s,r]
\end{align}
where
\begin{align}
    F_{1,j}[s,r] & := \left\{(1-\epp{20})\mathbb{E}[X_j'[s,r]] < X_j'[s,r]\right\},\\
    F_{2,j}[s,r] & := \left\{X_j[s,r] < (1+\epp{20})\mathbb{E}[X_j[s,r]]\right\},\\
    F_{3,j}[s,r] & :=\left\{ (1-\epp{20})\mathbb{E}[Y_j'[s,r]] < Y_j'[s,r]\right\},\\
    F_{4,j}[s,r] & := \left\{ Z_j[s,r] < \mathbb{E}[Z_j[s,r]] + \epp{20}\mathbb{E}[X_j'[s,r]]\right\}.
\end{align}
\end{definition}
Note that for $0\le j \le m$ and $r\ge T$, $X_j[r], Y_j[r], X'_j[r], Y'_j[r]$  and $Z_j[r]$ are identically distributed as $X[r], Y[r], X'[r], Y'[r]$  and $Z[r]$ in bitcoin protocol. Also, $F_j[s,r]$ is defined in the same manner as $F[s,r]$. Thus the proposer blockchain and all voter blockchains satisfy similar properties as those of the bitcoin blockchains.
% in Lemma \ref{lemma: asyn properties}.

\begin{comment}
\begin{lemma} \label{lemma: asyn properties j}
For all integers $T\le s < r-\frac{2}{q}$ and $0\le j \le m$, under $F_j[s,r]$, the following properties holds:
\begin{align} \label{equ: j X'>(1-delta/20)q(1-q)^T-1s}
    (1-\epp{20})q(1-q)^{T-1}(r-s) < X'_j[s,r]
\end{align}
\begin{align} \label{equ: j X <(1+epp20)q(r-s)}
    X_j[s,r] < (1+\epp{20})q(r-s)
\end{align}
\begin{align}
    (1-\epp{3})q(r-s) <Y'_j[s,r], \label{a j 3.0}
\end{align}
\begin{align}
    Z_j[s,r]<(1-\xi)\frac{q}{1-q}(r-s) + \epp{20}q(1-q)^{T-1}(r-s) \le (1-\epq{2}{3})q(r-s), \label{equ:j  asyn Z<(1-2/3beta)qs}
\end{align}
\begin{align} \label{equ: asyn j Z<(1-delta/2)X'}
    Z_j[s,r+T] < (1-\epp{2})X'_j[s,r]
\end{align}
\begin{align} \label{equ: j Z<Y'}
    Z_j[s-T,r+T] < Y'_j[s,r]
\end{align}
\end{lemma}
\begin{proof}
For $j= 0,1,\ldots, m$, the lemma admits essentially the same proof at that for Lemma \ref{lemma: asyn properties}.
\end{proof}
\end{comment}

\begin{definition}
For all integers $T\le s < r-\frac{2}{q}$ and $0\le j \le m$, define blockchain $j$'s typical event with respect to $[s,r]$ as
\begin{align}
    J_j[s, r] := \cap_{0\le a \le s-T, b\ge 0}F_j[s-a,r+b].
\end{align}
\end{definition}

\begin{lemma} \label{lemma: prob of event Jj}
For all integers $T\le s < r -\frac{2}{q}$ and $0\le j \le m$,
\begin{align}
    P(J_j[s,r]) > 1 - 5\eta'^{-2}e^{-\eta'(r-s)}
\end{align}
where $\eta'$ is defined in \eqref{def: gamma'}.
\end{lemma}
\begin{proof}
For $j= 0,1,\ldots, m$, the lemma admits essentially the same proof at that for Lemma \ref{lemma: asyn prob. typical event}.
\end{proof}

\begin{lemma} \label{lemma: asyn pre blockchain growth j}
Let $ T \le s \le r-T$  and $0\le j\le m$ be integers. Suppose an honest blockchain is of length $l$ by round $s$. Then by round $r$, the length of every honest voter blockchain is at least $l + X'_j[s,r-T+1]$.
\end{lemma}
\begin{proof}
For $j= 0,1,\ldots, m$, the lemma admits essentially the same proof at that for Lemma \ref{lemma: asyn pre blockchain growth}.
\end{proof}

\begin{lemma} \label{lemma: asyn blockchain growth j}
For all integers $T \le s < r -\frac{2}{q}$, $k\ge 2q(r-s)$ and $0\le j \le m$, under typical event $J_j[s,r]$, every honest miner's $k$-deep block of blockchain $j$ by round $r$ must be mined before round $s$.
\end{lemma}
\begin{proof}
For $j= 0,1,\ldots, m$, the lemma admits essentially the same proof at that for Lemma \ref{lemma: asyn blockchain growth}.
\end{proof}

\begin{theorem} \label{thm: asyn blockchain growth j}
Let $r, s, s_1, j$ be integers satisfying $T \le s_1 \le s < r-\frac{2}{q}$ and $0\le j\le m$.  Then under typical event $J_j[s, r-T]$, the length of every honest miner's blockchain $j$ must increase by at least $(1-\epp{10})(1-q)^{T}q(r-s_1)$ during rounds $\{s_1, \ldots, r\}$.
\end{theorem}
\begin{proof}
For $j= 0,1,\ldots, m$, the theorem admits essentially the same proof at that for Theorem \ref{thm: asyn blockchain growth}.
\end{proof}

\begin{theorem} \label{thm: asyn blockchain quality j}
Let $r,s,k, j$ be integers satisfying $T \le s < r-\frac{2}{q}$, $k\ge 2q(r-s)$ and $1\le j \le m$. Suppose an honest miner's blockchain $j$ has more than $k$ blocks by round $r$. Under event $J_j[s, r-T]$, by round $r$, at least $\frac{\xi}{2}$ fraction of the last $k$ blocks of this miner's blockchain $j$ are honest.
\end{theorem}
\begin{proof}
For $j= 1,\ldots, m$, the theorem admits essentially the same proof at that for Theorem \ref{thm: asyn blockchain quality}.
\end{proof}

\begin{theorem} \label{thm: asyn common prefix j}
Let $r, s ,k, j$ be integers satisfying $T \le s < r-\frac{2}{q}$, $k> 2q(r-s)$ and $1\le j \le m$. If by round $r$ an honest miner's blockchain $j$ has a $k$-deep prefix, then the prefix is permanent after round $r$ under
$J_j[s+T, r-T]$.
\end{theorem}
\begin{proof}
For $j= 1,\ldots, m$, the theorem admits essentially the same proof at that for Theorem \ref{thm: asyn common prefix}.
\end{proof}

Define
\begin{align} \label{def: asyn delta k}
    \delta_k = 5 m (\eta')^{-2} e^{-\eta' \frac{k}{2q}+(2T+1)\eta'} .
 % \delta_k =5m\eta'^{-2}e^{-\eta' \frac{k}{2}+(2T+1)q\eta'}
\end{align}

Recall that $R_l$ denotes the round in which the first proposer block on level $l$ is mined (Definition~\ref{d:Rlek}).

\begin{lemma}\label{lemma: asyn one honest block fix all}
Consider a given level $l$. Let $k$ be a positive integer satisfying $k\ge 5$. If by some round $r > \max \left\{ \frac{k}{2q}, R_l + T\right\}$,  every voter blockchain contains at least one honest block mined after round $R_l$ which is at least $k$-deep, then $\textbf{LedSeq}_l(r)$ is $\delta_k$-permanent after round $r$. %, where $\delta_k$ is defined in \eqref{def: asyn delta k}.
\end{lemma}
\begin{proof}
Let
\begin{align} \label{a 9.-25}
  s = r-\left\lfloor \frac{k}{2q}\right\rfloor,
\end{align}
which must be a positive integer because $2qr>k$. Define
\begin{align}
    J = \cap_{j=1,2,\ldots,m}J_j[s+T, r-T].
\end{align}

For $j=1,\dots,m$, let $B_j$ denote an honest block on an honest voter blockchain $j$ which is mined after round $R_l$ and is at least $k$-deep by round $r$.  Since $r-s = \left\lfloor \frac{k}{2q}\right\rfloor > \frac{2}{q}$, according to Theorem \ref{thm: asyn common prefix j}, $B_j$ and its ancestors are permanent after round $r$ under $J_j[s+T, r-T]$  . Hence, $B_1,\ldots,B_m$ and all their ancestors must be permanent after round $r$ under $J$. Thus, all voter blockchains' voting are permanent. Since $B_1,\ldots,B_m$ are honest, they would have voted for all levels up to level $l$ of the proposer blockchain by the voting rule. Hence, the leader block sequence up to level $l$ is permanent after round $r$ under $J$.
Note that
\begin{align}
P(J) & = 1-P(\cup_{j=1,2,\ldots,m}J^c_j[s+T, r-T]) \\
& \ge 1 - \sum_{j=1}^m P(J^c_j[s+T, r-T])\label{a 9.-1}\\
& = 1-mP(J^c_1[s+T,r-T]) \label{a 9.-05}\\
& > 1-  {\color{black}5m\eta'^{-2}}e^{-\eta'(r-s-2T)},  \label{a 9.0}
\end{align}
where \eqref{a 9.-1} is due to the union bound, \eqref{a 9.-05} is due to symmetry of all {\color{black}voter} blockchains, \eqref{a 9.0} is due to Lemma~\ref{lemma: prob of event Jj}. %lemma: asyn prob. typical event}.
By \eqref{a 9.-25}, {\color{black}we have} $2q(r-s+1)>k$, {\color{black}so that \eqref{a 9.0} becomes}
\begin{align}
P(G) & > 1-  {\color{black}5m\eta'^{-2}}e^{-\eta' \frac{k}{2q}+\eta' (2T+1)} \\
& = 1-\delta_k.
\end{align}
Thus, the leader block sequence up to level $l$ is $\delta_k$-permanent after round $r$.
\end{proof}

%-----------------------------------------------------------------------

\begin{lemma}\label{lemma: asyn get one honest block}
If positive integers $R$, $r$, and $k$ satisfy $k\ge 5$ and
\begin{align} \label{def: a r 8.1}
r \ge  \frac{2(k+1)}{(1-\frac{\xi}{10})\xi q(1-q)^{T}}+1,
\end{align}
then by round $R+r$, with probability of at least $1-\delta_k$, all honest voter blockchains have an honest block mined after $R$ which is at least $k$ deep.
\end{lemma}
\begin{proof}
Let
\begin{align}
    \ell = \left\lceil  \frac{2k}{\xi} \right\rceil. \label{a 8.111}
\end{align}
Let
\begin{align}
    s_1 = \left\lfloor \frac{k}{q\xi}\right\rfloor.  \label{a 7.95}
\end{align}
Then
\begin{align}
    \ell&\ge \frac{2k}{\xi} \\
    & \ge 2q \left\lfloor \frac{k}{q\xi}\right\rfloor \\
    & = 2qs_1. \label{a 8.105}
\end{align}
Obviously $r > \frac{2}{q}$. According to the Theorem \ref{thm: asyn blockchain growth j}, under event $J_j[R, R+r-T]$, an honest voter blockchain $j$'s growth during $\{R, R+1,\ldots, R+r-1\}$ is at least
\begin{align}
    (1-\frac{\xi}{10}) q(1-q)^{T}  r  & \ge \frac{2(k+1)}{\xi} \label{a 8.112} \\
    & > \ell, \label{a 8.113}
    % (1-\frac{\xi}{10})\xi q(1-q)^{T}  r  & \ge \frac{2k}{\xi}+1 \label{a 8.112} \\
    %     & > \ell, \label{a 8.113}
\end{align}
where \eqref{a 8.112} is due to \eqref{def: a r 8.1} and \eqref{a 8.113} is due to \eqref{a 8.111}.

Note that $s_1 = \left\lfloor \frac{k}{q\xi}\right\rfloor > \frac{2}{q}$.
According to Theorem \ref{thm: asyn blockchain quality j} and \eqref{a 8.105},  under event $J_j[R + r - s_1, R + r-T]$, at least $\epp{2}$ fraction of the last $\ell$ blocks of this voter blockchain $j$ are honest. Because $\frac{\xi}{2}\ell \ge k$, the earliest of these  honest blocks must be at least $k$ deep.

By \eqref{def: a r 8.1} and \eqref{a 7.95}, it is easy to see that $s_1\le r$.
Hence $ J_j[R+r-s_1, R+r-T]\subset J_j[R, R+r-T]$.
We define
\begin{align}
    J = \cap_{j=1,2,\ldots,m}J_j[R+r-s_1, R+r-T] .
\end{align}
Under event $J$, by round $R+r$, every honest voter blockchain has an honest block mined after round $R$ which is at least $k$ deep. The probability of the typical event can be lower bounded:
\begin{align}
P(J) & = P(\cap_{j=1,2,\ldots,m}J_j[R+r-s_1, R+r-T]) \\
& = 1-P(\cup_{j=1,2,\ldots,m}J^c_j[R+r-s_1, R+r-T])\label{a 8.0}\\
& \ge 1-mP(J^c_1[R+r-s_1, R+r-T]) \label{a 8.1}\\
&> 1- {\color{black}5}m\eta'^{-2}e^{-\eta' (s_1-T)} \label{a 8.2}
% \\ &= 1- 5m\eta'^{-2}e^{-\eta' s_1 + \eta'T}
\end{align}
where \eqref{a 8.1} is due to the union bound and symmetry of all voter blockchains and \eqref{a 8.2} is due to  Lemma \ref{lemma: prob of event Jj}. Moreover,
\begin{align}
    s_1 - T
    &= \lfloor \frac{k}{q\xi} \rfloor - T \\
    &> \frac{k}{2q} - 1 - 2T .
    % qs_1 & > q\left( \frac{k}{q\xi}-1\right) \label{a 8.21}\\
    % & = \frac{k}{\xi}-q\\
    % &{\color{black} > k - \frac{1}{6}} \\
    % & > \frac{k}{2},
\end{align}
% where \eqref{a 8.21} is due to \eqref{a 7.95}.
Therefore,
\begin{align}
    P(J)
    % &  > 1- {\color{black}5m\eta'^{-2}}e^{-\eta \frac{k}{2} + \eta'T} \\
    & > 1- \delta_k .
\end{align}

In summary, by round $R+r$, with probability at least $1-\delta_k$, all honest voter blockchains have an honest block mined after round $R$ which is at least $k$ deep.
\end{proof}
%-----------------------------------------------------------------------
\begin{theorem} \label{thm: aysn permanent leader block}
Fix $\epsilon \in (0,1)$. Let $R_l$ be {\color{black}the round during which the first proposer block on level $l$ is mined}. {\color{black}For every integer}
\begin{align} \label{def: asyn r thm permanent}
    r \ge  \frac{5}{(1-\frac{\xi}{10})\xi \eta'(1-q)^T}\left(\log\frac{10m\eta'^{-2}}{\epsilon} + \eta'(2T+1) \right),
\end{align}
{\color{black}the leader sequence up to level $l$} is $\epsilon-$permanent after round $R_l + r$.
\end{theorem}

\begin{proof}
Let
\begin{align}
k = \left\lceil \frac{2}{\eta'}\log \frac{10m\eta'^{-2}}{\epsilon} + 2(2T+1)q\right\rceil, \label{a 11.01}
\end{align}
and
\begin{align}
    s = \left\lceil \frac{2(k+1)}{(1-\frac{\xi}{10})\xi q(1-q)^T}+1  \right\rceil.\label{a 11.02}
\end{align}
Let $\delta_k$ be as defined as in \eqref{def: asyn delta k}.

According to Lemma \ref{lemma: asyn get one honest block} and \eqref{a 11.02}, by round $R_l + s$, all honest voter blockchains have an honest block which is mined after $R_l$ and is at least $k$ deep with probability at least {\color{black}$1-\delta_k$}. Under this event, according  to Lemma \ref{lemma: asyn one honest block fix all} ({\color{black}evidently, $R_l+s > \frac{k}{2q}$} and $s>\frac{2}{q}$),  the leader sequence up to level $l$ is $\delta_k$-permanent after round $R_l+s$. Therefore, {\color{black}the leader sequence up to level $l$} is $2\delta_k$-permanent after round $R_l+s$. Note that
\begin{align}
\delta_k & = {\color{black}5m\eta'^{-2}}e^{-\eta' \frac{k}{2}+\eta'(2T+1)} \\
& \le {\color{black}5m\eta'^{-2}}e^{-\log {\color{black}\frac{12m\eta'^{-2}}{\epsilon}}} \label{a 11.11}\\
& =  \frac{\epsilon}{2}.
\end{align}
{\color{black}the leader sequence up to level $l$} is $\epsilon$-permanent after round $R_l+s$.

From \eqref{a 11.01}, it is easy to verify that $k>10$. As a consequence, we have
\begin{align}
 s  & <   \frac{2(k+1)}{(1-\frac{\xi}{10})\xi q(1-q)^T}+2  \label{a 11.05}\\
 & = \frac{2k+2+2(1-\frac{\xi}{10})\xi q(1-q)^T}{(1-\frac{\xi}{10})\xi q(1-q)^T} \\
 & < \frac{\frac{5}{2}(k-1)}{(1-\frac{\xi}{10})\xi q(1-q)^T} \label{a 11.051}\\
  & < \frac{5}{(1-\frac{\xi}{10})\xi \eta'(1-q)^T}\left(\log\frac{10m\eta'^{-2}}{\epsilon} + \eta'(2T+1) \right)
  \label{a 11.065}\\
 & \le r, \label{a 11.066}
\end{align}
where \eqref{a 11.05} is due to \eqref{a 11.01},
{\color{black}\eqref{a 11.051} is due to $k>10$, \eqref{a 11.065} is due to \eqref{a 11.01},
and \eqref{a 11.066} is by \eqref{def: asyn r thm permanent}}.

Since $r>s$,  is $\epsilon$-permanent after round $R_l+r$ by Lemma \ref{lemma: permanent ever after}.
\end{proof}

\begin{theorem} \label{thm: asyn blockchain quality for proposer block}
(Blockchain quality theorem for proposer block for bounded-delay model)
Let $r,s,k$ be integers satisfying $T \le s < r-\frac{2}{q}$ and $k\ge 2q(r-s)$. Suppose an honest proposer blockchain has more than $k$ leader blocks {\color{black}by round $r$}. Under event $J_0[s, r-T]$, by round $r$, at least $\epp{2}$ fraction of the last $k$ leader blocks {\color{black}of the {\color{black}proposer} blockchain} are honest.
\end{theorem}

\begin{proof}
{\color{black}Let $l$ denote the highest level} of the proposer blockchain by round $r$. {\color{black}Evidently $l>k$.}
Let $l^*$ be the {\color{black}highest level} before $l-k+1$ {\color{black}on which} the first proposer block is honest. {\color{black}$l^*$ may be as high as $l-k$ and as low as $0$, which corresponds to the genesis block}.
Let $r^*$ be the round when the first block on level $l^*$ is mined. If this block is the genesis block, {\color{black}then} $r^* = 0$. If $r^*>0$, since blocks on level $l^*$ are more than $k$ blocks away from the last level by round $r$, {\color{black}we have} $r^* < s$ according to Lemma \ref{lemma: asyn blockchain  growth j}. In any cases, we have $[s,r]\subset [r^*+1,r]$.

{\color{black}Since the first proposer block on every level within $\{l^*+1, \dots, l-k\}$ is adversarial,}
from level $l^*+1$ to level $l$, there must be at least one adversarial block on every level except (possibly) on the levels between $l-k+1$ and $l$ where the leading block is honest. Let $x$ be the number of honest leader blocks on levels $\{l-k+1,\ldots, l\}$. Then during rounds $\{r^*+1,\ldots, r-1\}$, the total number of adversarial {\color{black}proposer} blocks is no fewer than $l-l^*-x$, i.e.,
\begin{align}
Z_0[r^*+1, r] & \ge l-l^* - x. \label{a 10.-2}
\end{align}
Under $J_0[s,r-T]$, $E_0[r^*+1, r-T]$ also occurs. Thus,
\begin{align}
x & \ge l - l^* - Z_0[r^*+1, r]\label{a 10.-1}\\
& > l-l^* - (1-\epp{2})X'_0[r^*+1, r-T] \label{a 10.0}\\
& \ge \epp{2} (l-l^*) \label{a 10.3}\\
& \ge \epp{2} k, \label{a 10.4}
\end{align}
where \eqref{a 10.0} is due to \eqref{equ: asyn Z<(1-delta/2)X'},  \eqref{a 10.3} is due to Lemma {\color{black}\ref{lemma: asyn pre blockchain growth j}}, and {\color{black}\eqref{a 10.4}} is due to $l-l^* \ge k$.
To sum up,  {\color{black}we have $x>\epp{2} k$ and the proof is complete}.
\end{proof}

\begin{theorem} \label{thm: asyn transaction becomes permanent}
For every $\epsilon >0$ and every integer
\begin{align} \label{def: asyn r permanent tx}
r \ge  \frac{25}{(1-\epp{10})^2\xi^2\eta'(1-q)^{2T}} \left(\log \frac{20m\eta'^{-2}}{\epsilon}+\eta'(2T+1)\right),
\end{align}
an honest transaction that enters into a block is $\epsilon$-permanent {\color{black}$r$ rounds after the block is broadcast}.
\end{theorem}
\begin{proof}
Let
\begin{align}
\ell & = {\color{black} \left\lceil (1-\epp{10})q(1-q)^Tr \right\rceil}\label{a 12.91} \\
k & = {\color{black}\left\lfloor \frac{\xi}{2}\ell \right\rfloor} \label{a 12.901}\\
w & = \left\lfloor \frac{\ell}{2q} \right\rfloor \label{a 12.92}\\
u & =  \left\lfloor \frac{k}{2q} \right\rfloor. \label{a 12.94}
\end{align}
{\color{black}
Let $R$ be the round during which the block including the honest transaction is broadcast.}
Define
\begin{align}
  J = J_0[R+r-u, R+r-T]  \cap J_0[R+r-w, R+r-T] \cap J_0[R, R+r-T].
\end{align}

Note that 1) According to Theorem \ref{thm: asyn blockchain growth j} and \eqref{a 12.91}, under $G_0[R,R+r-T]$, the proposer blockchain grows by at least $\ell$ leader blocks {\color{black}during rounds $\{R,\ldots, R+r\}$}. 2) According to Theorem \ref{thm: asyn blockchain quality for proposer block}, under event $G_0[R+r-w, R+r-T]$, by round $R+r$ the last $\ell$ leader blocks includes at least $\epp{2}$ fraction of honest ones. Since $k \le \epp{2}\ell$, at least $k$ out of the last $\ell$ leader blocks are honest. 3) According to Lemma \ref{lemma: asyn blockchain growth j}, under event
$G_0[R+r-u, R+r-T]$, the deepest one of these $k$ honest leader blocks is mined at least $\frac{k}{2q}$ rounds {\color{black}before} round $R+r$. 4) We have
\begin{align}
    \frac{k}{2q} & \ge \frac{1}{2q}\left\lfloor \frac{\xi}{2}\ell \right\rfloor \label{a 14.001}\\
    & \ge  \frac{1}{2q} \left\lfloor \frac{\xi}{2}(1-\epp{10})q(1-q)^{T}r \right\rfloor \label{a 14.002}\\
    & \ge  \frac{1}{2q} \left\lfloor \frac{25}{2(1-\epp{10})\xi\eta'(1-q)^{T}} \left(\log \frac{20m\eta'^{-2}}{\epsilon} + \eta'(2T+1) \right)
    \right\rfloor \label{a 14.003}\\
    & > \frac{1}{2q} \left(
 \frac{25}{2(1-\epp{10})\xi\eta'(1-q)^{T}} \left(\log \frac{20m\eta'^{-2}}{\epsilon} + \eta'(2T+1) \right)
    -1 \right) \label{a 14.004}\\
    & > \frac{1}{2q} \left(
    \frac{10}{(1-\epp{10})\xi\eta'(1-q)^{T}} \left(\log \frac{20m\eta'^{-2}}{\epsilon} + \eta'(2T+1) \right)
    \right) \label{a 14.005}\\
    & = \frac{5}{(1-\epp{10})\xi\eta'(1-q)^{T}} \left(\log \frac{10m\eta'^{-2}}{\frac{\epsilon}{2}} + \eta'(2T+1) \right)  \label{a 14.006},
\end{align}
where \eqref{a 14.001} is due to {\color{black}\eqref{a 12.901}}, \eqref{a 14.002} is due to {\color{black}\eqref{a 12.91}}, and \eqref{a 14.005} is obvious due to $\xi \in (0,1]$.
According to Theorem \ref{thm: aysn permanent leader block} and \eqref{a 14.006},  the deepest honest leader block is $\frac{\epsilon}{2}$-permanent after round $R+r$  under event $J$.
Next, we will lower bound probability of $J$.
Note that
\begin{align}
u   & \le \frac{k}{2q} \\
    & \le \frac{\xi \ell}{4q} \label{a 15.00}\\
    & < \frac{\ell}{2q} - 1 \label{a 15.01} \\
    & < w, \label{a 15.016}
\end{align}
where \eqref{a 15.00} is due to \eqref{a 12.901}, \eqref{a 15.01} is due to $q\le \ep$, and \eqref{a 15.016} is due to \eqref{a 12.92}. Also,
\begin{align}
w & \le \frac{\ell}{2q} \label{a 15.0171} \\
& \le  {\color{black}r}, \label{a 15.025}
\end{align}
where \eqref{a 15.0171} is due to \eqref{a 12.92} and
\eqref{a 15.025} is due to {\color{black}\eqref{a 12.91}}.
We have $u < w < s$. According to definition, $G_0[R+r-u, R+r-T] \subset  G_0[R+r-w, R+r-T] \subset G_0[R,R+r-T]$. Then,
\begin{align}
P(J) = & P \left( J_0[R+r-u, R+r-T]\right)  \\
> & 1- {\color{black}5\eta'^{-2}}e^{- \eta' u} \label{a 13.-1} \\
> & 1- {\color{black}5\eta^{-2}}e^{- \eta' (\frac{k}{2q}-1)} \label{a 13.0}\\
> & 1- {\color{black}5\eta'^{-2}}e^{-\log \frac{10\eta'^{-2}}{\epsilon}}\label{a 13.1} \\
= & 1- \frac{\epsilon}{2},
\end{align}
where \eqref{a 13.-1} is due to Lemma \ref{lemma: prob of event Jj},
\eqref{a 13.0} is due to \eqref{a 12.94},
\eqref{a 13.1} is due to \eqref{a 14.006}.
According to the union rule, the deepest honest leader block is $\epsilon$-permanent after round $R+r$.
{\color{black}According to Lemma \ref{lemma: honest leader block include blocks}, the honest transaction will become a $\epsilon$-permanent transaction after round $R+r$.}
\end{proof}

\section{Conclusion}
%%[DG
In this paper, we have analyzed
% We discard the finite horizon assumption in the analyses of
the bitcoin backbone protocol and the Prism backbone protocol using more general models than previously seen in the literature.
In particular, we allow the blockchains to have unlimited lifespan and allow the block propagation delays to be arbitrary but bounded.
% A typical event is defined for each interval such that the probability of the occurrence of ``atypical'' events is exponentially decreasing with interval length regardless of whether the lifespan of blockchain is finite or not.
Under %this definition, we prove
the new setting,
we rigorously establish a blockchain growth property, %: under some typical event, the length of every honest miner's blockchain must increase by at least $(1-\ep)qr$ and at most $2qr$ during a interval of length $r$, where $q$ is the probability that at least one honest block is mined during each round.
%We also prove
a blockchain quality property, %: under some typical event, at least $\frac{\xi}{2}$ fraction of the last $k$ blocks of an honest miner's blockchain are honest,
and a common prefix property %: an honest miner's $k$-deep prefix is permanent after round $r$ under some typical events.
for the bitcoin backbone protocol.
%Following
Under this framework, we have also proved a blockchain growth property and a blockchain quality property of the leader sequence in the Prism protocol.
%  a common prefix property of the leader sequence is not straightforward.
% we have proved a similar result:
We have also shown that the leader sequence is %$\epsilon$-
permanent with high probability after sufficient amount of wait time.
% enough number of rounds. At last, we show that
As a consequnce, every honest transaction will eventually enter the final ledger and become permanent with probability higher than $1-\epsilon$
% , where the
%number of rounds need to be waited
after a confirmation time proportional to security parameter $\log\frac{1}{\epsilon}$.
This paper provide explicit bounds for the bitcoin and the Prism backbone protocols, which furthers understanding of both protocols and provides practical guidance to public transaction ledger protocol design.
%%DG]
%==============================================================

\bibliographystyle{ieeetr}
\bibliography{ref}

\end{document}